\documentclass[11pt]{article}
\usepackage{amssymb, amsmath}
\usepackage{epsfig}
\usepackage{times}
\usepackage{hyperref}
\usepackage{graphicx}

\newtheorem{theorem}{Theorem}[section]

\newtheorem{definition}[theorem]{Definition}

\newtheorem{claim}[theorem]{Claim}
\newtheorem{lemma}[theorem]{Lemma}

\newtheorem{corollary}[theorem]{Corollary}

\newcommand{\qedsymb}{\hfill{\rule{2mm}{2mm}}}
\newenvironment{proof}[1][]{\begin{trivlist}
\item[\hspace{\labelsep}{\bf\noindent Proof#1:\/}] }{\qedsymb\end{trivlist}}

\def\calA{{\cal A}}

\def\calP{{\cal P}}

\def\calV{{\cal V}}
\def\Z{{\mathbb{Z}}}
\def\C{\mathbb{C}}
\def\R{\mathbb{R}}
\def\T{\mathbb{T}}

\def\mod{{\rm{mod}}}

\def\diam{\mbox{diam}}
\def\poly{{\rm{poly}}}

\def\dist{{\rm{dist}}}

\def\Re{\mathop{\rm Re}}
\newcommand\expo[1]{{{\mathrm{exp}}\left(#1\right)}}
\newcommand{\Exp}{\mathop{\mathrm{Exp}}}          
\def\YES{{\rm YES}}
\def\NO{{\rm NO}}

\newcommand\set[2]{\left\{ #1 \left|\; #2 \right. \right\}}

\newcommand\floor[1]{{\lfloor #1 \rfloor}}

\newcommand\round[1]{{\lfloor #1 \rceil}}

\newcommand\abs[1]{{\left| {#1} \right|}}
\newcommand\ip[1]{{\langle {#1} \rangle}}

\newcommand{\CVP}{{\mbox{\sc CVP}}}     

\newcommand{\GapCVP}{\mbox{\sc GapCVP}}     
\newcommand{\coGapCVP}{\mbox{\sc coGapCVP}}
\newcommand{\GapSVP}{\mbox{\sc GapSVP}}     
\newcommand{\SIVP}{\mbox{\sc SIVP}}
\newcommand{\GIVP}{\mbox{\sc GIVP}}
\newcommand{\SVP}{\mbox{\sc SVP}}
\newcommand{\LWE}{\mbox{\sc LWE}}

\newcommand{\DGS}{\mbox{\sc DGS}}
\newcommand{\NP}{\mbox{\sc NP}}

\newcommand\ket[1]{{ |{#1} \rangle }}

\renewcommand{\vec}[1]{\mathbf{#1}}    

\newcommand{\eps}{\varepsilon}
\renewcommand{\epsilon}{\varepsilon}

\begin{document}

\title{\bf On Lattices, Learning with Errors, \\ Random Linear Codes, and Cryptography}

\author{
 Oded Regev \footnote{School of Computer Science, Tel Aviv University, Tel Aviv 69978, Israel. Supported
   by an Alon Fellowship, by the Binational Science Foundation, by the Israel Science Foundation, by the Army Research Office grant DAAD19-03-1-0082,
   by the European Commission under the Integrated Project QAP funded by the IST directorate as Contract Number 015848, and by a European Research Council (ERC) Starting Grant.
    } }


\maketitle

\begin{abstract}
Our main result is a reduction from worst-case lattice problems such as $\GapSVP$ and $\SIVP$ to a certain learning
problem. This learning problem is a natural extension of the `learning from parity with error' problem to higher
moduli. It can also be viewed as the problem of decoding from a random linear code. This, we believe, gives a strong
indication that these problems are hard. Our reduction, however, is quantum. Hence, an efficient solution to the
learning problem implies a {\em quantum} algorithm for $\GapSVP$ and $\SIVP$. A main open question is whether this
reduction can be made classical (i.e., non-quantum).

We also present a (classical) public-key cryptosystem whose security is based on the hardness of the learning problem.
By the main result, its security is also based on the worst-case quantum hardness
of $\GapSVP$ and $\SIVP$. The new cryptosystem is much more efficient than previous
lattice-based cryptosystems: the public key is of size $\tilde{O}(n^2)$ and encrypting a message increases its size by a factor of $\tilde{O}(n)$
(in previous cryptosystems these values are $\tilde{O}(n^4)$ and $\tilde{O}(n^2)$, respectively). In fact, under the
assumption that all parties share a random bit string of length $\tilde{O}(n^2)$, the size of the public key can be
reduced to $\tilde{O}(n)$.

\end{abstract}

\section{Introduction}

\paragraph{Main theorem.}
For an integer $n \ge 1$ and a real number $\eps \ge 0$, consider the `learning from parity with error' problem,
defined as follows: the goal is to find an unknown $\vec s \in \Z_2^n$ given a list of `equations with errors'
\begin{align*}
 \ip{\vec s, \vec a_1} &\approx_\eps b_1 ~(\mod ~ 2) \\
 \ip{\vec s, \vec a_2} &\approx_\eps b_2 ~(\mod ~ 2) \\
 &\vdots
\end{align*}
where the $\vec a_i$'s are chosen independently from the uniform distribution on $\Z_2^n$, $\ip{\vec s, \vec
a_i}=\sum_j s_j (a_i)_j$ is the inner product modulo $2$ of $\vec s$ and $\vec a_i$, and each equation
is correct independently with probability $1-\eps$. More precisely, the input to the problem consists of
pairs $(\vec a_i, b_i)$ where each $\vec a_i$ is chosen independently and uniformly from $\Z_2^n$
and each $b_i$ is independently chosen to be equal to $\ip{\vec s, \vec a_i}$
with probability $1-\eps$. The goal is to find $\vec s$. Notice
that the case $\eps = 0$ can be solved efficiently by, say, Gaussian elimination. This requires $O(n)$ equations and
$\poly(n)$ time.

The problem seems to become significantly harder when we take any positive $\eps > 0$. For example, let us consider
again the Gaussian elimination process and assume that we are interested in recovering only the first bit of $\vec s$. Using
Gaussian elimination, we can find a set $S$ of $O(n)$ equations such that $\sum_S \vec a_i$ is $(1,0,\ldots,0)$.
Summing the corresponding values $b_i$ gives us a guess for the first bit of $\vec s$. However, a standard calculation
shows that this guess is correct with probability $\frac{1}{2} + 2^{-\Theta(n)}$. Hence, in order to obtain the first
bit with good confidence, we have to repeat the whole procedure $2^{\Theta(n)}$ times. This yields an algorithm that
uses $2^{O(n)}$ equations and $2^{O(n)}$ time. In fact, it can be shown that given only $O(n)$ equations, the $\vec s'
\in \Z_2^n$ that maximizes the number of satisfied equations is with high probability $\vec s$. This yields a simple
maximum likelihood algorithm that requires only $O(n)$ equations and runs in time $2^{O(n)}$.

Blum, Kalai, and Wasserman \cite{BlumKW03} provided the first subexponential algorithm for this problem. Their
algorithm requires only $2^{O(n/\log n)}$ equations/time and is currently the best known algorithm for the problem. It
is based on a clever idea that allows to find a small set $S$ of equations (say, $O(\sqrt{n})$) among $2^{O(n/\log n)}$
equations, such that $\sum_S \vec a_i$ is, say, $(1,0,\ldots,0)$. This gives us a guess for the first bit of $\vec s$
that is correct with probability $\frac{1}{2} + 2^{-\Theta(\sqrt{n})}$. We can obtain the correct value with high
probability by repeating the whole procedure only $2^{O(\sqrt{n})}$ times. Their idea was later shown to have
other important applications, such as the first $2^{O(n)}$-time algorithm for solving the shortest vector problem \cite{KumarS01,AjtaiKS01}.

An important open question is to explain the apparent difficulty in finding efficient algorithms for this learning
problem. Our main theorem explains this difficulty for a natural extension of this problem to higher moduli, defined
next.

Let $p=p(n) \le \poly(n)$ be some prime integer and consider a list of `equations with error'
\begin{align*}
 \ip{\vec s, \vec a_1} &\approx_\chi b_1 ~(\mod ~ p) \\
 \ip{\vec s, \vec a_2} &\approx_\chi b_2 ~(\mod ~ p) \\
 &\vdots
\end{align*}
where this time $\vec s \in \Z_p^n$, $\vec a_i$ are chosen independently and uniformly from $\Z_p^n$, and $b_i \in
\Z_p$. The error in the equations is now specified by a probability distribution $\chi:\Z_p \to \R^+$ on $\Z_p$.
Namely, for each equation $i$, $b_i = \ip{\vec s, \vec a_i} + e_i$ where each $e_i \in \Z_p$ is chosen independently
according to $\chi$. We denote the problem of recovering $\vec s$ from such equations by $\LWE_{p,\chi}$ (learning with
error). For example, the learning from parity problem with error $\eps$ is the special case where $p=2$,
$\chi(0)=1-\eps$, and $\chi(1)=\eps$. Under a reasonable assumption on $\chi$ (namely, that $\chi(0)>1/p+1/\poly(n)$),
the maximum likelihood algorithm described above solves $\LWE_{p,\chi}$ for $p \le \poly(n)$ using $\poly(n)$ equations
and $2^{O(n \log n)}$ time. Under a similar assumption, an algorithm resembling the one by Blum et al. \cite{BlumKW03}
requires only $2^{O(n)}$ equations/time. This is the best known algorithm for the $\LWE$ problem.

Our main theorem shows that for certain choices of $p$ and $\chi$, a solution to $\LWE_{p,\chi}$ implies a quantum
solution to worst-case lattice problems.

\begin{theorem}[Informal]\label{thm:mainthminformal} Let $n,p$ be integers and $\alpha \in (0,1)$ be such that $\alpha p > 2\sqrt{n}$.
If there exists an efficient algorithm that solves $\LWE_{p,\bar{\Psi}_\alpha}$ then there exists an efficient quantum
algorithm that approximates the decision version of the shortest vector problem ($\GapSVP$) and the shortest independent vectors problem ($\SIVP$)
to within $\tilde{O}(n/\alpha)$ in the worst case.
\end{theorem}

The exact definition of $\bar{\Psi}_\alpha$ will be given later. For now, it is enough to know that it is a
distribution on $\Z_p$ that has the shape of a discrete Gaussian centered around 0 with standard deviation $\alpha p$,
as in Figure~\ref{fig:discretized}. Also, the probability of $0$ (i.e., no error) is roughly $1/(\alpha p)$. A possible
setting for the parameters is $p = O(n^2)$ and $\alpha = 1/(\sqrt{n} \log^2 n)$ (in fact, these are the parameters that
we use in our cryptographic application).

\begin{figure}[h]
 \center{\epsfxsize=2.5in\epsfbox{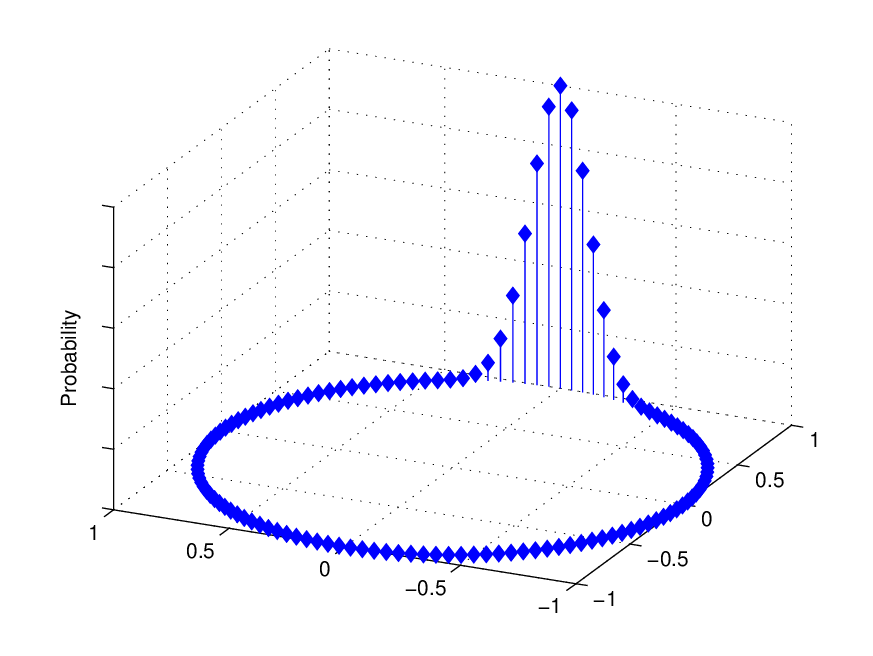} \epsfxsize=2.5in\epsfbox{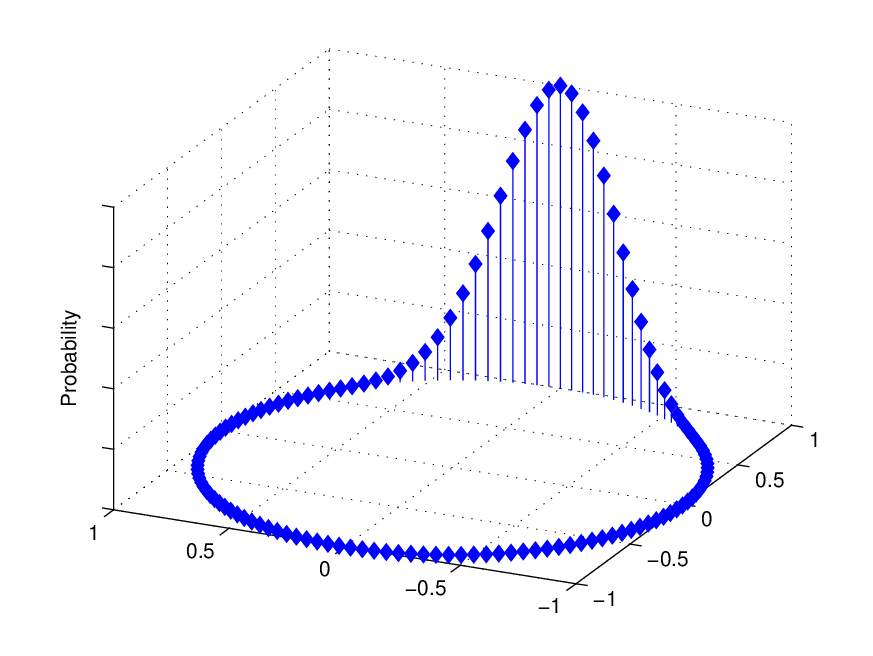}}
  \caption{$\bar{\Psi}_\alpha$ for $p=127$ with $\alpha=0.05$ (left) and $\alpha=0.1$ (right). The elements of $\Z_p$ are arranged on a circle.}
  \label{fig:discretized}
\end{figure}

$\GapSVP$ and $\SIVP$ are two of the main computational problems on lattices. In $\GapSVP$, for instance,
the input is a lattice, and the goal is to approximate the length of the shortest nonzero lattice vector.
The best known polynomial time algorithms for them yield only mildly subexponential approximation factors \cite{LLL,SchnorrLLL,AjtaiKS01}. It is conjectured that there is no classical
(i.e., non-quantum) polynomial time algorithm that approximates them to within any polynomial factor. Lattice-based constructions of
one-way functions, such as the one by Ajtai \cite{Ajtai96}, are based on this conjecture.

One might even conjecture that there is no {\em quantum} polynomial time
algorithm that approximates $\GapSVP$ (or $\SIVP$) to within any polynomial factor. One can then interpret the main
theorem as saying that based on this conjecture, the $\LWE$ problem is hard. The only evidence supporting this
conjecture is that there are no known quantum algorithms for lattice problems that outperform classical
algorithms, even though this is probably one of the most important open questions in the field of quantum computing.\footnote{If forced to make a guess,
the author would say that the conjecture is true.}

In fact, one could also interpret our main theorem as a way to disprove this conjecture: if one finds an efficient
algorithm for $\LWE$, then one also obtains a quantum algorithm for approximating worst-case lattice problems. Such a
result would be of tremendous importance on its own. Finally, we note that it is possible that our
main theorem will one day be made classical. This would make all our results stronger and the above discussion unnecessary.

The $\LWE$ problem can be equivalently presented as the problem of decoding random linear codes. More specifically, let
$m=\poly(n)$ be arbitrary and let $\vec s \in \Z_p^n$ be some vector. Then, consider the following problem: given a
random matrix $Q \in \Z_p^{m \times n}$ and the vector $\vec t = Q \vec s + \vec e \in \Z_p^m$ where each coordinate of
the error vector $\vec e \in \Z_p^m$ is chosen independently from $\bar{\Psi}_\alpha$, recover $\vec s$. The Hamming
weight of $\vec e$ is roughly $m(1-1/(\alpha p))$ (since a value chosen from $\bar{\Psi}_\alpha$ is $0$ with
probability roughly $1/(\alpha p)$). Hence, the Hamming distance of $\vec t$ from $Q\vec s$ is roughly $m(1-1/(\alpha
p))$. Moreover, it can be seen that for large enough $m$, for any other word $\vec s'$, the Hamming distance of $\vec
t$ from $Q \vec s'$ is roughly $m(1-1/p)$. Hence, we obtain that approximating the nearest codeword problem to within
factors smaller than $(1-1/p)/(1-1/(\alpha p))$ on random codes is as hard as quantumly approximating worst-case
lattice problems. This gives a partial answer to the important open question of understanding the hardness of decoding
from random linear codes.

It turns out that certain problems, which are seemingly easier than the $\LWE$ problem, are in fact equivalent to the
$\LWE$ problem. We establish these equivalences in Section~\ref{sec:reductions} using elementary reductions. For example, being able to distinguish a
set of equations as above from a set of equations in which the $b_i$'s are chosen uniformly from $\Z_p$ is equivalent
to solving $\LWE$. Moreover, it is enough to correctly distinguish these two distributions for some non-negligible
fraction of all $\vec s$. The latter formulation is the one we use in our cryptographic applications.

\paragraph{Cryptosystem.}
In Section~\ref{sec:cryptosystem} we present a public key cryptosystem
and prove that it is secure based on the hardness of the
$\LWE$ problem. We use the standard security notion of semantic, or IND-CPA, security (see, e.g., \cite[Chapter 10]{KatzL08}).
The cryptosystem and its security proof are entirely classical.
In fact, the cryptosystem itself is quite simple; the reader is encouraged to glimpse
at the beginning of Section~\ref{sec:cryptosystem}.
Essentially, the idea is to provide a list of equations as above as the public key;
encryption is performed by summing some of the equations (forming another
equation with error) and modifying the right hand side depending on the
message to be transmitted. Security follows from the fact that a list of equations with error is
computationally indistinguishable from a list of equations in which the $b_i$'s are chosen
uniformly.

By using our main theorem, we obtain that the security of the
system is based also on the worst-case quantum hardness of approximating $\SIVP$ and $\GapSVP$ to
within $\tilde{O}(n^{1.5})$. In other words, breaking our cryptosystem implies
an efficient quantum algorithm for approximating $\SIVP$ and $\GapSVP$ to within $\tilde{O}(n^{1.5})$.
Previous cryptosystems, such as the Ajtai-Dwork cryptosystem \cite{AjtaiDwork} and the one by Regev \cite{Regev03A},
are based on the worst-case (classical) hardness of the unique-$\SVP$ problem, which can be related
to $\GapSVP$ (but not $\SIVP$) through the recent result of Lyubashevsky and Micciancio~\cite{LyubashevskyM09}.

Another important feature of our cryptosystem is its improved efficiency. In previous cryptosystems, the public key
size is $\tilde{O}(n^4)$ and the encryption increases the size of messages by a factor of $\tilde{O}(n^2)$. In our
cryptosystem, the public key size is only $\tilde{O}(n^2)$ and encryption increases the size of messages by a factor of
only $\tilde{O}(n)$. This possibly makes our cryptosystem practical. Moreover, using an idea of Ajtai
\cite{AjtaiHardLattices}, we can reduce the size of the public key to $\tilde{O}(n)$. This requires all users of the
cryptosystem to share some (trusted) random bit string of length $\tilde{O}(n^2)$. This can be achieved by, say, distributing
such a bit string as part of the encryption and decryption software.

We mention that learning problems similar to ours were already suggested as possible sources of cryptographic
hardness in, e.g., \cite{BlumFKL94,Alekhnovich03}, although this was done without establishing any
connection to lattice problems. In another related work \cite{AjtaiHardLattices}, Ajtai suggested
a cryptosystem that has several properties in common with ours (including its efficiency),
although its security is not based on worst-case lattice problems.

\paragraph{Why quantum?} This paper is almost entirely classical. In fact, quantum is needed only in one step
in the proof of the main theorem. Making this step classical would make the entire reduction classical. To demonstrate
the difficulty, consider the following situation. Let $L$ be some lattice and let $d=\lambda_1(L)/n^{10}$ where
$\lambda_1(L)$ is the length of the shortest nonzero vector in $L$. We are given an oracle that for any point $\vec
x\in \R^n$ within distance $d$ of $L$ finds the closest lattice vector to $\vec x$. If $\vec x$ is not within distance
$d$ of $L$, the output of the oracle is undefined. Intuitively, such an oracle seems quite powerful; the best known
algorithms for performing such a task require exponential time. Nevertheless, we do not see any way to use this oracle
classically. Indeed, it seems to us that the only way to generate inputs to the oracle is the following: somehow choose
a lattice point $\vec y\in L$ and let $\vec x = \vec y + \vec z$ for some perturbation vector $\vec z$ of length at
most $d$. Clearly, on input $\vec x$ the oracle outputs $\vec y$. But this is useless since we already know $\vec y$!

It turns out that quantumly, such an oracle is quite useful. Indeed, being able to compute $\vec y$ from $\vec x$
allows us to {\em uncompute} $\vec y$. More precisely, it allows us to transform the quantum state $\ket{\vec x, \vec
y}$ to the state $\ket{\vec x, 0}$ in a reversible (i.e., unitary) way. This ability to erase the contents of a memory
cell in a reversible way seems useful only in the quantum setting.

\paragraph{Techniques.}
Unlike previous constructions of lattice-based public-key
cryptosystems, the proof of our main theorem uses an `iterative construction'. Essentially, this means
that instead of `immediately' finding very short vectors in a lattice, the reduction proceeds in steps where in each
step shorter lattice vectors are found. So far, such iterative techniques have been used only in the construction
of lattice-based one-way functions \cite{Ajtai96, CaiNerurkar97atow,
Micciancio04perfect, MicciancioR04}. Another novel aspect of our main theorem is its crucial use of quantum
computation. Our cryptosystem is the first {\em classical} cryptosystem whose security is based on a
{\em quantum} hardness assumption (see \cite{MooreRZ07} for a somewhat related recent work).

Our proof is based on the Fourier transform of Gaussian measures, a technique that was developed in previous papers
\cite{Regev03A, MicciancioR04, AharonovR04}. More specifically, we use a parameter known as the smoothing parameter, as
introduced in \cite{MicciancioR04}. We also use the discrete Gaussian distribution and approximations to its Fourier
transform, ideas that were developed in \cite{AharonovR04}.

\paragraph{Open questions.} The main open question raised by this work is whether Theorem~\ref{thm:mainthminformal} can be dequantized:
can the hardness of $\LWE$ be established based on the classical
hardness of $\SIVP$ and $\GapSVP$? We see no reason why this should be
impossible. However, despite our efforts over the last few years, we were not able to show this. As mentioned above, the difficulty is that
there seems to be no classical way to use an oracle that solves the closest vector problem within small distances.
Quantumly, however, such an oracle turns out to be quite useful.

Another important open question is to determine the hardness of the learning from parity with errors problem (i.e., the
case $p=2$). Our theorem only works for $p > 2\sqrt{n}$. It seems that in order to prove similar results for smaller
values of $p$, substantially new ideas are required. Alternatively, one can interpret our inability to prove hardness
for small $p$ as an indication that the problem might be easier than believed.

Finally, it would be interesting to relate the $\LWE$ problem to other average-case problems in the literature,
and especially to those considered by Feige in \cite{Feige02}. See Alekhnovich's paper \cite{Alekhnovich03}
for some related work.

\paragraph{Followup work.}
We now describe some of the followup work that has appeared
since the original publication of our results in 2005~\cite{Regev05}.

One line of work focussed on improvements to our cryptosystem.
First, Kawachi, Tanaka, and Xagawa~\cite{KawachiTX07}
proposed a modification to our cryptosystem that slightly improves the encryption blowup to $O(n)$, essentially getting
rid of a $\log$ factor. A much more significant improvement is
described by Peikert, Vaikuntanathan, and Waters in~\cite{PeikertVW07}.
By a relatively simple modification to the cryptosystem, they managed
to bring the encryption blowup down to only $O(1)$, in addition
to several equally significant improvements in running time.
Finally, Akavia, Goldwasser, and Vaikuntanathan~\cite{AkaviaGV09}
show that our cryptosystem remains secure even if almost the entire
secret key is leaked.

Another line of work focussed on the design of other cryptographic protocols
whose security is based on the hardness of the $\LWE$ problem.
First, Peikert and Waters \cite{PeikertW07} constructed,
among other things, CCA-secure cryptosystems (see also~\cite{Peikert09} for a simpler construction).
These are cryptosystems that are secure even if the adversary is allowed access
to a decryption oracle (see, e.g., \cite[Chapter 10]{KatzL08}).
All previous lattice-based cryptosystems
(including the one in this paper) are not CCA-secure.
Second, Peikert, Vaikuntanathan, and Waters \cite{PeikertVW07} showed
how to construct oblivious transfer protocols, which are
useful, e.g., for performing secure multiparty computation.
Third, Gentry, Peikert, and Vaikuntanathan \cite{GentryPV08} constructed an identity-based
encryption (IBE) scheme. This is a public-key encryption scheme in which
the public key can be any unique identifier of the user;
very few constructions of such schemes are known.
Finally, Cash, Peikert, and Sahai \cite{CashPS09} constructed
a public-key cryptosystem that remains secure even when the
encrypted messages may depend upon the secret key.
The security of all the above constructions is based on the $\LWE$ problem
and hence, by our main theorem, also on the worst-case quantum hardness
of lattice problems.

The $\LWE$ problem has also been used by Klivans and Sherstov to show hardness
results related to learning halfspaces \cite{KlivansS06}.
As before, due to our main theorem, this implies hardness of learning halfspaces
based on the worst-case quantum hardness of lattice problems.

Finally, we mention two results giving further evidence for the hardness
of the $\LWE$ problem. In the first, Peikert~\cite{Peikert07}
somewhat strengthens our main theorem by replacing our worst-case
lattice problems with their analogues for the $\ell_q$ norm, where $2 \le q \le \infty$ is arbitrary.
Our main theorem only deals with the standard $\ell_2$ versions.

In another recent result, Peikert~\cite{Peikert09} shows that
the quantum part of our proof can be removed, leading to
a \emph{classical} reduction from $\GapSVP$ to the $\LWE$ problem.
As a result, Peikert is able to show that public-key cryptosystems
(including many of the above LWE-based schemes)
can be based on the classical hardness of $\GapSVP$, resolving
a long-standing open question (see also~\cite{LyubashevskyM09}).
Roughly speaking, the way Peikert circumvents the difficulty we
described earlier is by noticing that the \emph{existence}
of an oracle that is able to recover $\vec y$ from $\vec y + \vec z$, where
$\vec y$ is a random lattice point and $\vec z$ is a random perturbation
of length at most $d$, is by itself a useful piece of information
as it provides a lower bound on the length of the shortest nonzero vector.
By trying to construct such oracles for several different values
of $d$ and checking which ones work, Peikert is able to obtain
a good approximation of the length of the shortest nonzero vector.

Removing the quantum part, however, comes at a cost:
the construction can no longer be iterative, the hardness
can no longer be based on $\SIVP$, and even for hardness
based on $\GapSVP$, the modulus $p$ in the $\LWE$ problem must be exponentially big
unless we assume the hardness of a non-standard variant of $\GapSVP$.
Because of this, we believe that dequantizing our main theorem
remains an important open problem.

\subsection{Overview}

In this subsection, we give a brief informal overview of the proof of our main theorem, Theorem~\ref{thm:mainthminformal}.
The complete proof appears in Section~\ref{sce:main_theorem}. We do not discuss here
the reductions in Section~\ref{sec:reductions} and the cryptosystem in Section~\ref{sec:cryptosystem} as these parts
of the paper are more similar to previous work.

In addition to some very basic definitions related to lattices,
we will make heavy use here of the \emph{discrete Gaussian distribution on $L$ of width $r$}, denoted $D_{L,r}$. This is the
distribution whose support is $L$ (which is typically a lattice), and in which the probability of each $\vec x \in L$
is proportional to $\expo{-\pi \| \vec x/r\|^2}$ (see Eq.~\eqref{eq:discrete_gaussian} and Figure~\ref{fig:dist_d}).
We also mention here the \emph{smoothing parameter} $\eta_\epsilon(L)$. This is a real positive
number associated with any lattice $L$ ($\epsilon$ is an accuracy parameter which we can safely ignore here).
Roughly speaking, it gives the smallest $r$ starting from which $D_{L,r}$ `behaves like' a continuous
Gaussian distribution. For instance, for $r \ge \eta_\epsilon(L)$, vectors chosen from $D_{L,r}$ have norm
roughly $r \sqrt{n}$ with high probability. In contrast, for sufficiently small $r$, $D_{L,r}$ gives almost all its
mass to the origin $0$.
Although not required for this section, a complete list of definitions can be found in Section~\ref{sec:prelim}.

\begin{figure}[h]
 \center{\epsfxsize=2.5in\epsfbox{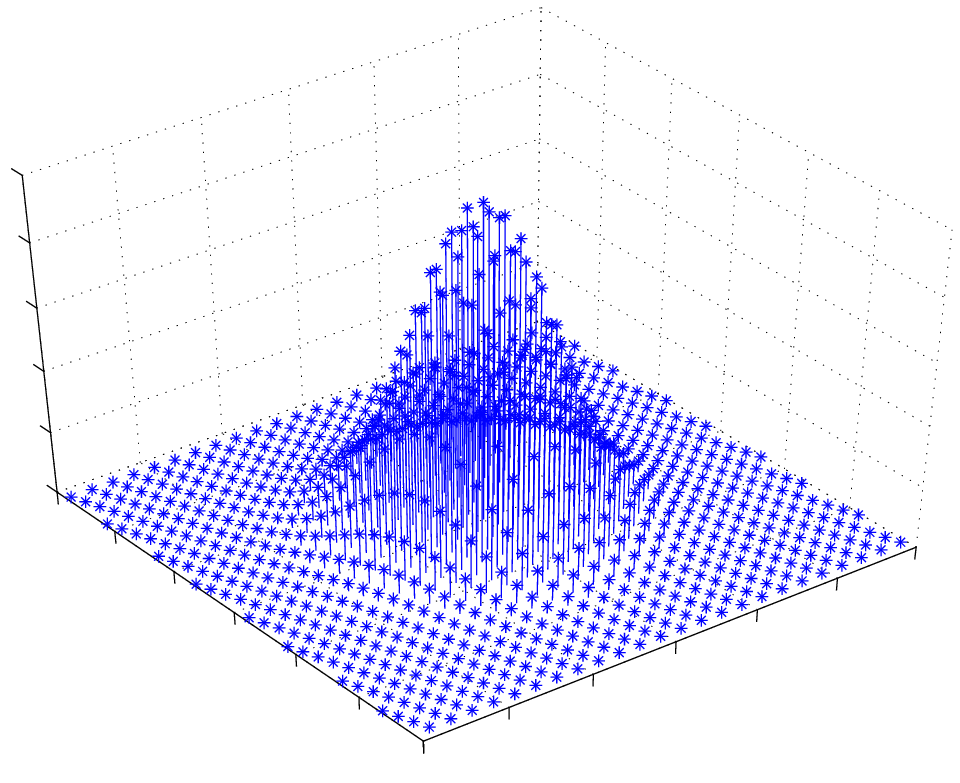}
         \epsfxsize=2.5in\epsfbox{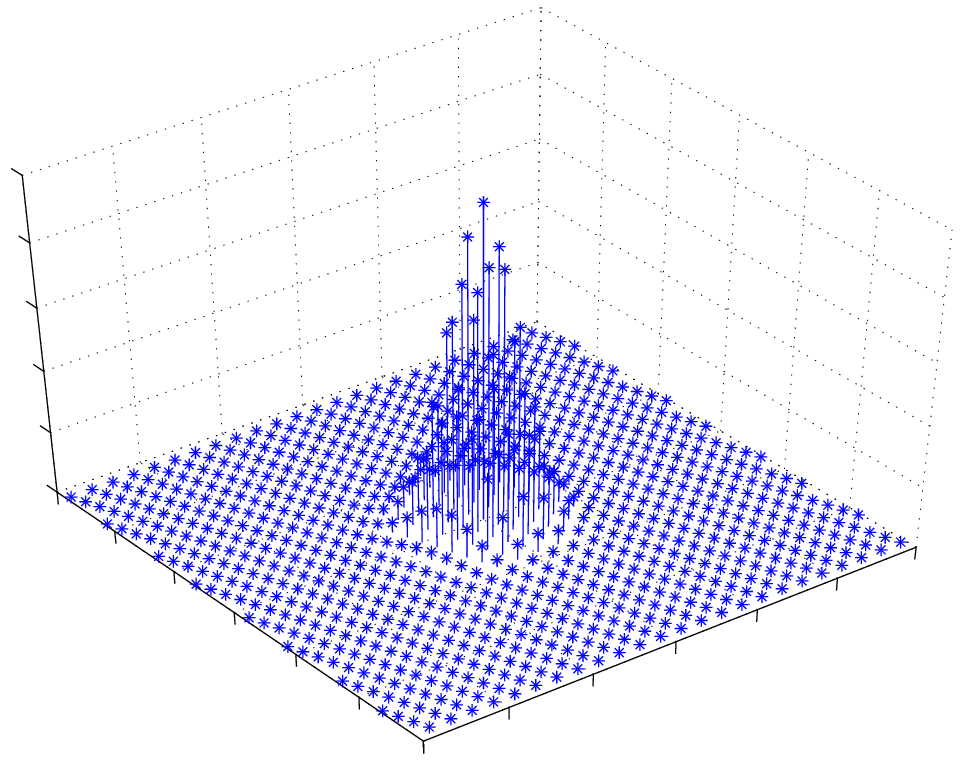}}
  \caption{$D_{L,2}$ (left) and $D_{L,1}$ (right) for a two-dimensional lattice $L$. The $z$-axis represents probability.}
  \label{fig:dist_d}
\end{figure}

Let $\alpha,p,n$ be such that $\alpha p > 2\sqrt{n}$, as required in Theorem~\ref{thm:mainthminformal},
and assume we have an oracle that solves $\LWE_{p,\bar{\Psi}_\alpha}$. For concreteness,
we can think of $p=n^2$ and $\alpha=1/n$. Our goal is to show
how to solve the two lattice problems mentioned in Theorem~\ref{thm:mainthminformal}.
As we prove in Subsection~\ref{ssec:standardlatticeproblems} using standard reductions, it suffices
to solve the following \emph{discrete Gaussian sampling problem} ($\DGS$):
Given an $n$-dimensional lattice $L$ and a number $r \ge \sqrt{2n} \cdot \eta_\eps(L) /\alpha$, output a sample
from $D_{L,r}$. Intuitively, the connection to $\GapSVP$ and $\SIVP$ comes from the fact that by taking
$r$ close to its lower limit $\sqrt{2n} \cdot \eta_\eps(L) /\alpha$, we can obtain short lattice vectors
(of length roughly $\sqrt{n} r$).
In the rest of this subsection we describe our algorithm for sampling from $D_{L,r}$.
We note that the exact lower bound on $r$ is not that important for purposes of this overview,
as it only affects the approximation factor we obtain for $\GapSVP$ and $\SIVP$. It suffices
to keep in mind that our goal is to sample from $D_{L,r}$ for $r$ that is rather small,
say within a polynomial factor of $\eta_\eps(L)$.

The core of the algorithm is the following procedure, which we call the `iterative step'.
Its input consists of a number $r$ (which is guaranteed to be not too small, namely,
greater than $\sqrt{2} p \eta_\eps(L)$), and $n^c$ samples from $D_{L,r}$ where $c$
is some constant. Its output is a sample from the distribution $D_{L,r'}$
for $r' = r \sqrt{n}/(\alpha p)$. Notice that since $\alpha p > 2\sqrt{n}$, $r' < r/2$.
In order to perform this `magic' of converting vectors of norm $\sqrt{n} r$ into
shorter vectors of norm $\sqrt{n} r'$, the procedure of course needs to use the
$\LWE$ oracle.

Given the iterative step, the algorithm for solving $\DGS$ works as follows. Let $r_i$ denote $r \cdot (\alpha p / \sqrt{n})^i$.
The algorithm starts by producing $n^c$ samples from $D_{L,r_{3n}}$.
Because $r_{3n}$ is so large, such samples can be computed efficiently by a simple
procedure described in Lemma~\ref{lem:bootstrap}.
Next comes the core of the algorithm: for $i = 3n, 3n-1,\ldots,1$ the algorithm uses its $n^{c}$
samples from $D_{L,r_i}$ to produce $n^{c}$ samples from $D_{L,r_{i-1}}$
by calling the iterative step $n^c$ times. Eventually, we end up with $n^{c}$ samples
from $D_{L,r_0} = D_{L,r}$ and we complete the algorithm by simply outputting
the first of those. Note the following crucial fact: using $n^{c}$ samples from $D_{L,r_i}$, we are able
to generate the same number of samples $n^{c}$ from $D_{L,r_{i-1}}$ (in fact, we could even generate more than $n^{c}$
samples). The algorithm would not work if we could only generate, say, $n^{c}/2$ samples, as this would
require us to start with an exponential number of samples.

We now finally get to describe the iterative step.
Recall that as input we have $n^c$ samples from $D_{L,r}$ and we are supposed to generate a sample from
$D_{L, r'}$ where $r' = r \sqrt{n}/(\alpha p)$. Moreover, $r$ is known and guaranteed to be at least
$\sqrt{2} p \eta_\eps(L)$, which can be shown to imply that $p/r < \lambda_1(L^*)/2$.
As mentioned above, the exact lower bound on $r$ does not matter much for this overview;
it suffices to keep in mind that $r$ is sufficiently larger than $\eta_\eps(L)$, and that $1/r$ is
sufficiently smaller than $\lambda_1(L^*)$.

The iterative step is obtained by combining two parts (see Figure~\ref{fig:iterative}).
In the first part, we construct a classical algorithm that uses the given samples and
the $\LWE$ oracle to solve the following closest vector problem, which we denote by $\CVP_{L^*, \alpha p/r}$:
given any point $\vec x\in \R^n$ within distance $\alpha p/r$ of the dual lattice $L^*$,
output the closest vector in $L^*$ to $\vec x$.\footnote{In fact, we only solve
$\CVP_{L^*, \alpha p/(\sqrt{2}r)}$ but for simplicity we ignore the factor $\sqrt{2}$ here.}
By our assumption on $r$, the distance between any two points in $L^*$ is greater than
$2\alpha p/r$ and hence the closest vector is unique.
In the second part, we use this algorithm
to generate samples from $D_{L,r'}$. This part is quantum (and in fact, the only quantum part of our proof).
The idea here is to use the $\CVP_{L^*, \alpha p/r}$ algorithm to generate a certain quantum superposition
which, after applying the quantum Fourier transform and performing a measurement, provides us
with a sample from $D_{L,r \sqrt{n}/(\alpha p)}$.
In the following, we describe each of the two parts in more detail.

\begin{figure}[h]
 \center{\epsfxsize=2.5in\epsfbox{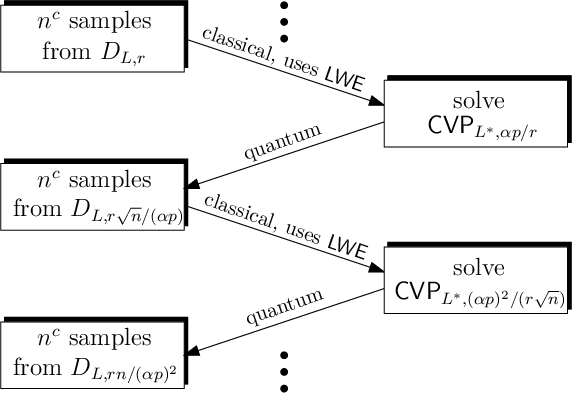}}
  \caption{Two iterations of the algorithm}
  \label{fig:iterative}
\end{figure}

\paragraph{Part 1:}
We start by recalling the main idea in \cite{AharonovR04}. Consider some
probability distribution $D$ on some lattice $L$ and consider its Fourier transform $f: \R^n \to \C$, defined as
$$ f(\vec x) = \sum_{\vec y \in L} D(\vec y) \expo{2\pi i \ip{\vec x, \vec y}} =
   \Exp_{\vec y \sim D}[\expo{2\pi i \ip{\vec x, \vec y}}]$$
where in the second equality we simply rewrite the sum as an expectation.
By definition, $f$ is $L^*$-periodic, i.e., $f(\vec x)=f(\vec x+ \vec y)$ for any $\vec x \in \R^n$ and $\vec y
\in L^*$. In \cite{AharonovR04} it was shown that given a polynomial number of samples from $D$, one can compute an
approximation of $f$ to within $\pm 1/\poly(n)$. To see this, note that by the Chernoff-Hoeffding bound,
if $\vec y_1,\ldots, \vec y_N$ are $N=\poly(n)$ independent samples from $D$, then
$$ f(\vec x) \approx \frac{1}{N} \sum_{j=1}^N \expo{2\pi i \ip{\vec x, \vec y_j}} $$
where the approximation is to within $\pm 1/\poly(n)$ and holds with probability
exponentially close to $1$, assuming that $N$ is a large enough polynomial.

By applying this idea to the samples from $D_{L,r}$ given to us as input,
we obtain a good approximation of the Fourier transform of $D_{L,r}$, which
we denote by $f_{1/r}$. It can be shown that since $1/r \ll \lambda_1(L^*)$ one
has the approximation
\begin{equation}\label{eq:fourier_tra_1}
f_{1/r}(\vec x) \approx \expo{-\pi (r \cdot \dist(L^*, \vec x))^2}
\end{equation}
(see Figure~\ref{fig:gaussianperiod1}).
Hence, $f_{1/r}(\vec x) \approx 1$ for any $\vec x \in L^*$ (in fact an equality holds)
and as one gets away from $L^*$, its value decreases. For points
within distance, say, $1/r$ from the lattice, its value is still some positive constant (roughly $\expo{-\pi}$).
As the distance from $L^*$ increases, the value of the function soon becomes negligible.
Since the distance between any two vectors in $L^*$ is at least
$\lambda_1(L^*) \gg 1/r$,
the Gaussians around each point of $L^*$ are well-separated.

\begin{figure}[h]
 \center{\includegraphics[width=8cm]{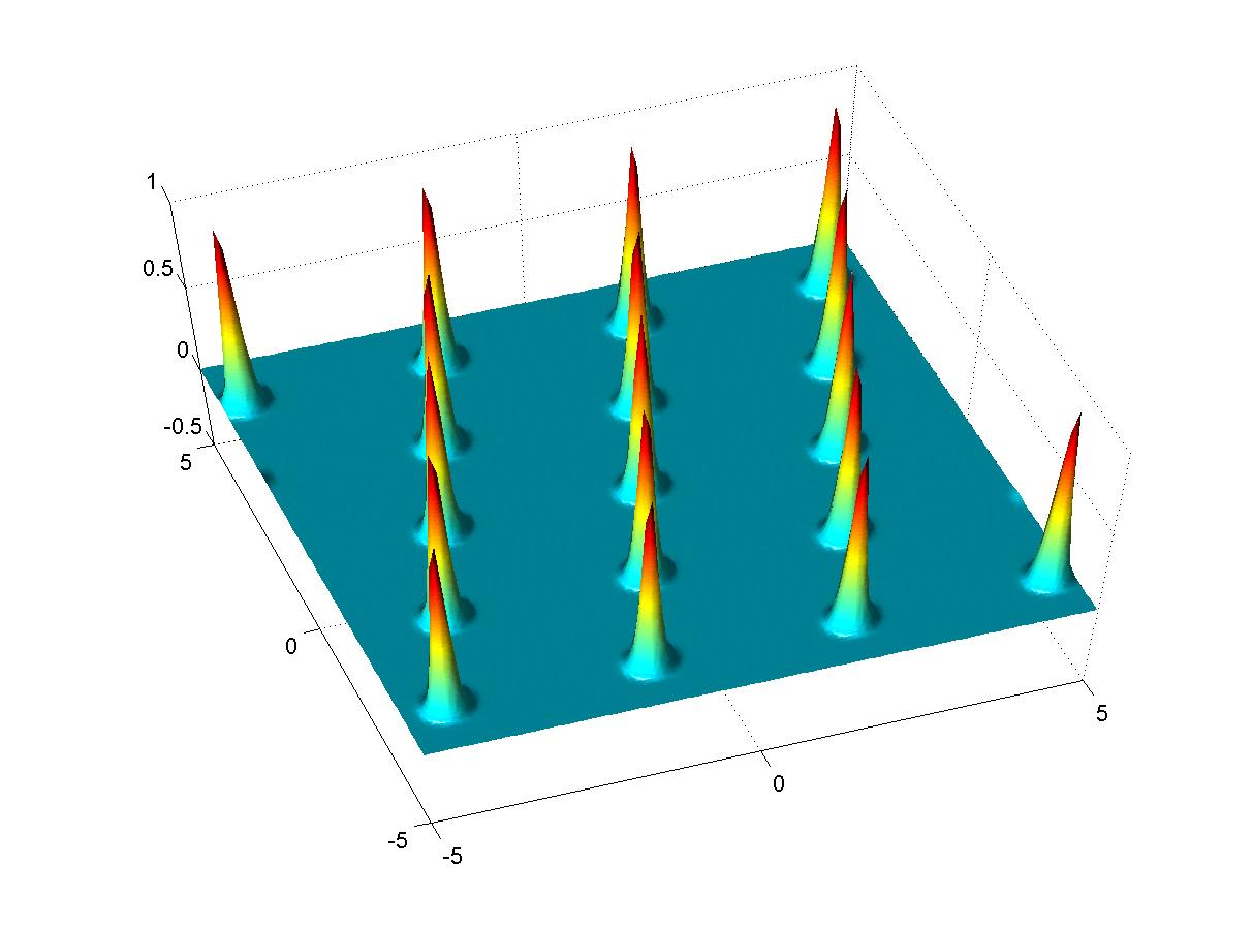}}
  \caption{$f_{1/r}$ for a two-dimensional lattice}
  \label{fig:gaussianperiod1}
\end{figure}

Although not needed in this paper, let us briefly outline how one can solve $\CVP_{L^*, 1/r}$ using samples from
$D_{L,r}$. Assume that we are given
some point $\vec x$ within distance $1/r$ of $L^*$. Intuitively, this $\vec x$ is located on one of the Gaussians of
$f_{1/r}$. By repeatedly computing an approximation of $f_{1/r}$ using the samples from $D_{L,r}$ as described above,
we `walk uphill' on $f_{1/r}$ in an attempt to find its `peak'. This peak corresponds to the
closest lattice point to $\vec x$. Actually, the procedure as described here does not quite work: due to the error in our approximation of
$f_{1/r}$, we cannot find the closest lattice point exactly. It is possible to overcome this difficulty; see~\cite{LiLyMi06}
for the details. The same procedure actually works for slightly longer distances, namely $O(\sqrt{\log n}/r)$,
but beyond that distance the value of $f_{1/r}$ becomes negligible and no useful information
can be extracted from our approximation of it.

Unfortunately, solving $\CVP_{L^*, 1/r}$ is not useful for the iterative step as it would lead
to samples from $D_{L,r \sqrt{n}}$, which is a wider rather than a narrower distribution
than the one we started with. This is not surprising, since our solution to $\CVP_{L^*, 1/r}$
did not use the $\LWE$ oracle. Using the $\LWE$ oracle, we will now show that we can gain an
extra $\alpha p$ factor in the radius, and obtain the desired $\CVP_{L^*, \alpha p /r}$
algorithm.

Notice that if we could somehow obtain samples from $D_{L,r/p}$ we would be done:
using the procedure described above, we could solve $\CVP_{L^*, p /r}$,
which is better than what we need. Unfortunately, it is not clear how to
obtain such samples, even with the help of the $\LWE$ oracle.
Nevertheless, here is an obvious way to obtain something similar to samples from $D_{L,r/p}$:
just take the given samples from $D_{L,r}$ and divide them by $p$.
This provides us with samples from $D_{L/p,r/p}$ where $L/p$ is the lattice $L$ scaled
down by a factor of $p$. In the following we will show how to use these samples
to solve $\CVP_{L^*, \alpha p /r}$.

Let us first try to understand what the distribution $D_{L/p,r/p}$
looks like. Notice that the lattice $L/p$ consists of $p^n$
translates of the original lattice $L$. Namely, for each
$\vec a \in \Z_p^n$, consider the set
$$ L + L \vec a / p = \{ L \vec b / p  ~|~ \vec b \in \Z^n,~~\vec b ~\mod~ p = \vec a \}.$$
Then $\{ L + L \vec a/p ~|~ \vec a \in \Z_p^n \}$ forms a partition of $L/p$.
Moreover, it can be shown that since $r/p$ is larger than the smoothing parameter $\eta_\epsilon(L)$,
the probability given to each $L+ L \vec a/p$ under $D_{L/p,r/p}$ is essentially the same, that is, $p^{-n}$.
Intuitively, beyond the smoothing parameter, the Gaussian
measure no longer `sees' the discrete structure of $L$, so in particular
it is not affected by translations (this will be shown in Claim~\ref{clm:gaussianweightshiftinvariant}).

This leads us to consider the following distribution, call it $\tilde{D}$.
A sample from $\tilde{D}$ is a pair $(\vec a, \vec y)$ where $\vec y$ is sampled from $D_{L/p,r/p}$,
and $\vec a \in \Z_p^n$ is such that $\vec y \in L+L\vec a/p$. Notice
that we can easily obtain samples from $\tilde{D}$ using the given samples
from $D_{L,r}$.
From the above discussion we have that the marginal distribution of $\vec a$ is essentially uniform. Moreover, by definition we have
that the distribution of $\vec y$ conditioned on any $\vec a$ is $D_{L+L\vec a/p, r/p}$.
Hence, $\tilde{D}$ is essentially identical to the distribution on pairs $(\vec a, \vec y)$
in which $\vec a\in \Z_p^n$ is chosen uniformly at random, and
then $\vec y$ is sampled from $D_{L+L\vec a/p, r/p}$. From now on,
we think of $\tilde{D}$ as being this distribution.

We now examine the Fourier transform of $D_{L+L\vec a/p, r/p}$ (see Figure
\ref{fig:gaussianperiod2}). When $\vec a$ is zero, we already know that the
Fourier transform is $f_{p/r}$. For general $\vec a$, a standard
calculation shows that the Fourier transform of $D_{L+L\vec a/p, r/p}$ is given by
\begin{equation}\label{eq:fourier_tra_2}
\expo{2 \pi i \ip{\vec a, \tau(\vec x)}/p} \cdot f_{p/r}(\vec x)
\end{equation}
where $\tau(\vec x) \in \Z_p^n$ is defined as
$$\tau(\vec x) := (L^*)^{-1} \kappa_{L^*}(\vec x) ~\mod~ p,$$
and $\kappa_{L^*}(\vec x)$ denotes the (unique) closest vector in $L^*$ to $\vec x$.
In other words, $\tau(\vec x)$ is the vector of coefficients of the vector
in $L^*$ closest to $\vec x$ when represented in the basis of $L^*$, reduced
modulo $p$. So we see that the Fourier transform $D_{L+L\vec a/p, r/p}$
is essentially $f_{p/r}$, except that each `hill' gets its own phase
depending on the vector of coefficients of the lattice point in its center.
The appearance of these phases is as a result of a well-known
property of the Fourier transform, saying that translation is transformed
to multiplication by phase.

\begin{figure}[h]
 \center{\includegraphics[width=8cm]{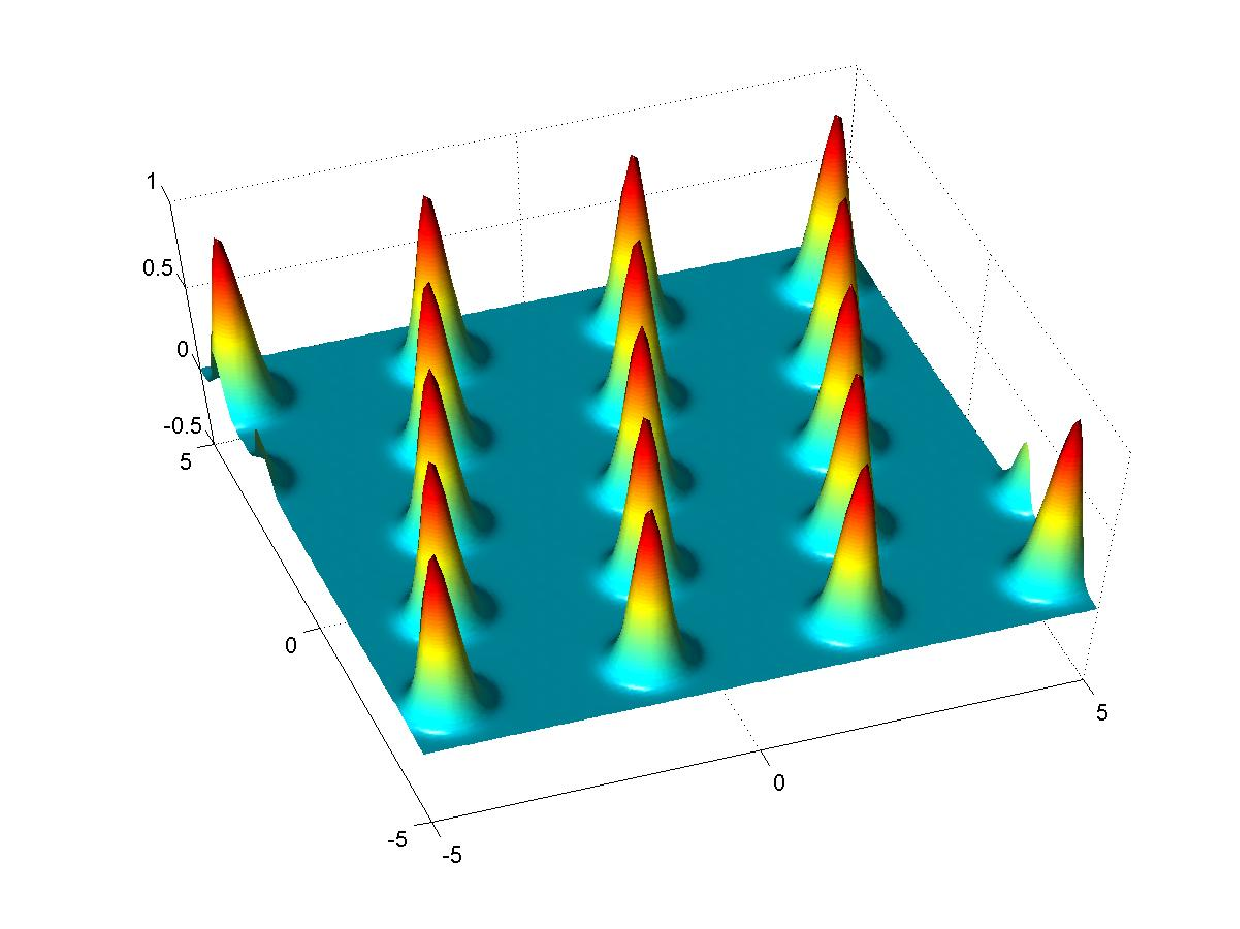} \includegraphics[width=8cm]{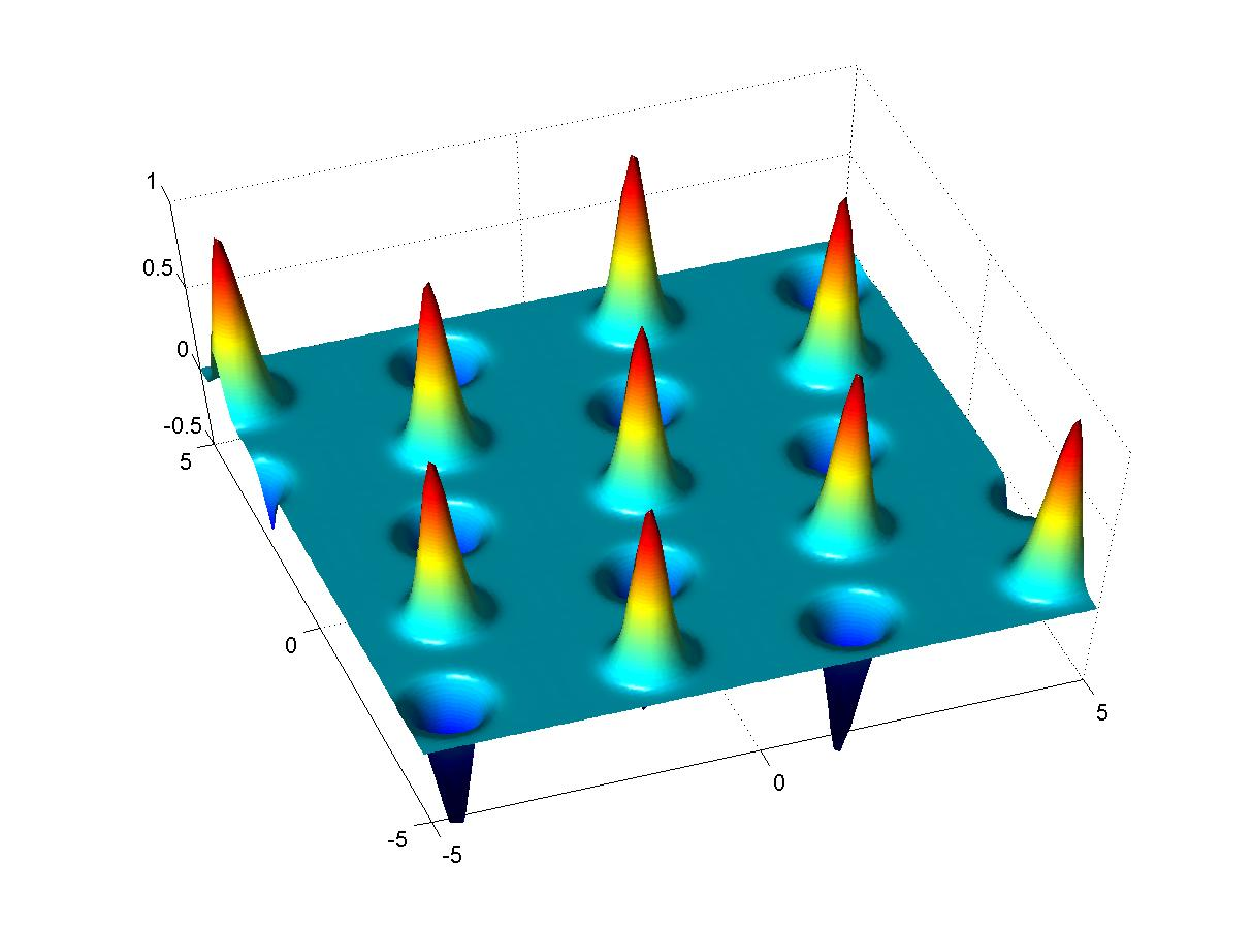}}
  \caption{The Fourier transform of $D_{L+L\vec a/p, r/p}$ with $n=2$, $p=2$, $\vec a=(0,0)$ (left), $\vec a=(1,1)$ (right).}
  \label{fig:gaussianperiod2}
\end{figure}

Equipped with this understanding of the Fourier transform of $D_{L+L\vec a/p, r/p}$, we can
get back to our task of solving $\CVP_{L^*, \alpha p /r}$. By the definition of the
Fourier transform, we know that the average of $\expo{2\pi i \ip{\vec x, \vec y}}$ over
$\vec y\sim D_{L+L\vec a/p, r/p}$ is given by \eqref{eq:fourier_tra_2}. Assume for
simplicity that $\vec x \in L^*$ (even though in this case finding the closest
vector is trivial; it is simply $\vec x$ itself). In this case,
\eqref{eq:fourier_tra_2} is equal to $\expo{2 \pi i \ip{\vec a, \tau(\vec x)}/p}$. Since the absolute value of this
expression is $1$, we see that for such $\vec x$, the random variable $\ip{\vec x, \vec y} ~\mod~ 1$ (where $\vec y
\sim D_{L+L\vec a/p, r/p}$) must be deterministically equal to $\ip{\vec a, \tau(\vec x)}/p ~\mod~ 1$ (this fact can also be
seen directly). In other words, when $\vec x \in L^*$, each sample $(\vec a, \vec y)$ from $\tilde{D}$,
provides us with a linear equation
$$ \ip{\vec a, \tau(\vec x)} = p\ip{\vec x, \vec y} ~\mod~ p$$
with $\vec a$ distributed essentially uniformly in $\Z_p^n$.
After collecting about $n$ such equations, we can use Gaussian elimination
to recover $\tau(\vec x) \in \Z_p^n$. And as we shall show in Lemma~\ref{lem:cvpmodp} using a simple reduction,
the ability to compute $\tau(\vec x)$ easily leads to the ability to compute
the closest vector to $\vec x$.

We now turn to the more interesting case in which $\vec x$ is not in $L^*$, but only
within distance $\alpha p / r$ of $L^*$. In this case,
the phase of \eqref{eq:fourier_tra_2} is still equal to $\expo{2 \pi i \ip{\vec a, \tau(\vec x)}/p}$.
Its absolute value, however, is no longer $1$, but still quite close to $1$ (depending
on the distance of $\vec x$ from $L^*$).
Therefore, the random variable $\ip{\vec x, \vec y} ~\mod~ 1$, where $\vec y
\sim D_{L+L\vec a/p, r/p}$, must be typically quite close to $\ip{\vec a, \tau(\vec x)}/p ~\mod~ 1$
(since, as before, the average of $\expo{2\pi i \ip{\vec x, \vec y}}$ is given by \eqref{eq:fourier_tra_2}).
Hence, each sample $(\vec a, \vec y)$ from $\tilde{D}$, provides us with a linear equation with error,
$$ \ip{\vec a, \tau(\vec x)} \approx \lfloor p\ip{\vec x, \vec y} \rceil ~\mod~ p.$$
Notice that $p\ip{\vec x, \vec y}$ is typically not an integer and
hence we round it to the nearest integer.
After collecting a polynomial number of such equations, we call the $\LWE$ oracle
in order to recover $\tau(\vec x)$. Notice that $\vec a$ is distributed
essentially uniformly, as required by the $\LWE$ oracle.
Finally, as mentioned above, once we are able to compute
$\tau(\vec x)$, computing $\vec x$ is easy (this will be shown in Lemma~\ref{lem:cvpmodp}).

The above outline ignores one important detail: what is the error distribution
in the equations we produce? Recall that the $\LWE$ oracle is only guaranteed
to work with error distribution $\bar{\Psi}_\alpha$. Luckily, as we will
show in Claim~\ref{clm:gaussian_noise} and Corollary~\ref{cor:gaussian_noise_one_dim} (using a rather technical proof), if $\vec x$ is at distance $\beta p/r$
from $L^*$ for some $0 \le \beta \le \alpha$, then the error distribution
in the equations is essentially $\bar{\Psi}_\beta$. (In fact, in order
to get this error distribution, we will have to modify the procedure a bit
and add a small amount of normal error to each equation.)
We then complete
the proof by noting (in Lemma~\ref{lem:learning_smooth}) that an oracle for solving $\LWE_{p,\bar{\Psi}_\alpha}$
can be used to solve $\LWE_{p,\bar{\Psi}_\beta}$ for any $0 \le \beta \le \alpha$
(even if $\beta$ is unknown).

\paragraph{Part 2:} In this part, we describe a quantum algorithm that, using a $\CVP_{L^*, \alpha p/r}$ oracle,
generates one sample from $D_{L,r \sqrt{n}/(\alpha p)}$. Equivalently, we show how to produce a sample from $D_{L,r}$ given a $\CVP_{L^*, \sqrt{n}/r}$ oracle.
The procedure is essentially the following: first, by using the $\CVP$ oracle, create a quantum state corresponding to
$f_{1/r}$. Then, apply the quantum Fourier transform and obtain a quantum state corresponding to $D_{L,r}$. By
measuring this state we obtain a sample from $D_{L,r}$.

In the following, we describe this procedure in more detail. Our first goal is to create a quantum state corresponding
to $f_{1 / r}$. Informally, this can be written as
\begin{equation}\label{eq:quantum_state_1}
 \sum_{\vec x \in \R^n} f_{1 / r} \ket{\vec x}.
\end{equation}
This state is clearly not well-defined. In the actual procedure, $\R^n$ is replaced with some finite set (namely, all
points inside the basic parallelepiped of $L^*$ that belong to some fine grid). This introduces several technical
complications and makes the computations rather tedious. Therefore, in the present discussion, we opt to continue
with informal expressions as in \eqref{eq:quantum_state_1}.

Let us now continue our description of the procedure. In order to prepare the state in \eqref{eq:quantum_state_1}, we
first create the uniform superposition on $L^*$,
$$ \sum_{\vec x \in L^*} \ket{\vec x}. $$
(This step is actually unnecessary in the real procedure, since there we work
in the basic parallelepiped of $L^*$; but for the present discussion,
it is helpful to imagine that we start with this state.)
On a separate register, we create a `Gaussian state' of width $1/r$,
$$ \sum_{\vec z \in \R^n} \expo{-\pi \|r \vec z \|^2} \ket{\vec z}. $$
This can be done using known techniques. The combined state of the system can be written as
$$ \sum_{\vec x \in L^*, \vec z \in \R^n} \expo{-\pi \|r \vec z\|^2} \ket{\vec x, \vec z}. $$
We now add the first register to the second (a reversible operation), and obtain
$$ \sum_{\vec x \in L^*, \vec z \in \R^n} \expo{-\pi \|r \vec z \|^2} \ket{\vec x, \vec x + \vec z}. $$
Finally, we would like to {\em erase}, or {\em uncompute}, the first register to obtain
\begin{align*}
 \sum_{\vec x \in L^*, \vec z \in \R^n} \expo{-\pi \|r \vec z\|^2} \ket{\vec x + \vec z} \approx
 \sum_{\vec z \in \R^n} f_{1 /r}(\vec z) \ket{\vec z}.
\end{align*}
However, `erasing' a register is in general not a reversible operation. In order for it to be reversible, we need to
be able to compute $\vec x$ from the remaining register $\vec x+ \vec z$. This is precisely why we need the $\CVP_{L^*,
\sqrt{n}/r}$ oracle. It can be shown that almost all the mass of $\expo{-\pi \|r \vec z\|^2}$ is on $\vec z$ such that
$\| \vec z \| \le \sqrt{n}/r$. Hence, $\vec x + \vec z$ is within distance $\sqrt{n}/r$ of the lattice and the oracle
finds the closest lattice point, namely, $\vec x$. This allows us to erase the first register in a reversible way.

In the final part of the procedure, we apply the quantum Fourier transform. This yields the quantum state corresponding
to $D_{L, r}$, namely,
$$ \sum_{\vec y \in L} D_{L,r}(\vec y) \ket{\vec y}.$$
By measuring this state, we obtain a sample from the distribution $D_{L, r}$ (or in fact from
$D_{L, r}^2 = D_{L, r/\sqrt{2}}$ but this is a minor issue).

\section{Preliminaries}\label{sec:prelim}

In this section we include some notation that will be used throughout the paper. Most of the notation is standard.
Some of the less standard notation is: the Gaussian function $\rho$ (Eq.~\eqref{def:rho}), the Gaussian distribution $\nu$ (Eq.~\eqref{def:nu}),
the periodic normal distribution $\Psi$ (Eq.~\eqref{def:Psi}), the discretization of a distribution on $\T$ (Eq.~\eqref{def:discretiz}),
the discrete Gaussian distribution $D$ (Eq.~\eqref{eq:discrete_gaussian}), the unique
closest lattice vector $\kappa$ (above Lemma~\ref{lem:transference}), and the smoothing parameter $\eta$ (Definition~\ref{def:smoothingpara}).

\paragraph{General.}
For two real numbers $x$ and $y>0$ we define $x ~\mod~ y$ as
$x - \floor{x/y}y$. For $x\in \R$ we define $\round{x}$ as the integer closest to $x$ or, in case two such integers
exist, the smaller of the two. For any integer $p \ge 2$, we write $\Z_p$ for the cyclic group $\{0,1,\ldots, p-1\}$
with addition modulo $p$. We also write $\T$ for $\R/\Z$, i.e., the segment $[0,1)$ with addition modulo $1$.

We define a {\em negligible amount} in $n$ as an amount that is asymptotically smaller than $n^{-c}$ for any constant $c>0$.
More precisely, $f(n)$ is a negligible function in $n$ if $\lim_{n \to \infty} n^c f(n) = 0$ for any $c>0$. Similarly,
a non-negligible amount is one which is at least $n^{-c}$ for some $c>0$. Also, when we say that an expression is
exponentially small in $n$ we mean that it is at most $2^{-\Omega(n)}$. Finally, when we say that an expression (most
often, some probability) is exponentially close to $1$, we mean that it is $1-2^{-\Omega(n)}$.

We say that an algorithm $\calA$ with oracle access is a {\em distinguisher} between two distributions if its acceptance
probability when the oracle outputs samples of the first distribution and its acceptance probability when the oracle
outputs samples of the second distribution differ by a non-negligible amount.

Essentially all algorithms and reductions in this paper have an exponentially small error probability,
and we sometimes do not state this explicitly.

For clarity, we present some of our reductions in a model that allows operations on real numbers. It is possible to
modify them in a straightforward way so that they operate in a model that approximates real numbers up to an error of
$2^{-n^c}$ for arbitrary large constant $c$ in time polynomial in $n$.

Given two probability density functions $\phi_1,\phi_2$ on $\R^n$, we define the {\em statistical distance} between them as
$$ \Delta(\phi_1,\phi_2) := \int_{\R^n} |\phi_1(\vec x) - \phi_2(\vec x)| d \vec x$$
(notice that with this definition, the statistical distance ranges in $[0,2]$). A similar definition can be given for discrete
random variables. The statistical distance satisfies the triangle inequality, i.e., for any $\phi_1,\phi_2,\phi_3$,
$$ \Delta(\phi_1,\phi_3) \le \Delta(\phi_1, \phi_2) + \Delta(\phi_2,\phi_3).$$
Another important fact which we often use is that the statistical distance cannot increase
by applying a (possibly randomized) function $f$, i.e.,
\begin{equation*}\label{eq:statistical_dst}
\Delta(f(X),f(Y)) \le \Delta(X,Y),
\end{equation*}
see, e.g., \cite{MicciancioBook}. In particular, this implies that the acceptance probability of any algorithm on
inputs from $X$ differs from its acceptance probability on inputs from $Y$ by at most $\frac{1}{2}\Delta(X,Y)$ (the
factor half coming from the choice of normalization in our definition of $\Delta$).

\paragraph{Gaussians and other distributions.}
Recall that the {\em normal distribution} with mean 0 and variance $\sigma^2$ is the distribution on $\R$ given by the
density function $\frac{1}{\sqrt{2\pi}\cdot \sigma} \expo{-\frac{1}{2}(\frac{x}{\sigma})^2}$ where $\expo{y}$ denotes
$e^y$. Also recall that the sum of two independent normal variables with mean 0 and variances $\sigma_1^2$ and
$\sigma_2^2$ is a normal variable with mean 0 and variance $\sigma_1^2+\sigma_2^2$. For a vector $\vec{x}$ and any
$s>0$, let
\begin{align}\label{def:rho}
\rho_{s}(\vec{x}) := \expo{-\pi \|\vec x/s\|^2}
\end{align}
be a Gaussian function scaled by a factor of $s$. We denote $\rho_1$ by $\rho$.
Note that $\int_{\vec x\in \R^n} \rho_{s}(\vec x) d\vec{x} = s^n$.
Hence,
\begin{align}\label{def:nu}
\nu_s := \rho_s / s^n
\end{align}
is an $n$-dimensional probability density function and as before,
we use $\nu$ to denote $\nu_1$. The dimension $n$ is implicit. Notice
that a sample from the Gaussian distribution $\nu_s$ can be obtained by taking $n$ independent samples from the $1$-dimensional Gaussian
distribution. Hence, sampling from $\nu_s$ to within arbitrarily good accuracy can be performed efficiently by using standard techniques.
For simplicity, in this paper we assume that we can sample from
$\nu_s$ exactly.\footnote{In practice, when only finite precision is available, $\nu_s$ can be approximated by picking
a fine grid, and picking points from the grid with probability approximately proportional to $\nu_s$. All our arguments
can be made rigorous by selecting a sufficiently fine grid.}
Functions are extended to sets in the usual way; i.e., $ \rho_{s}(A)  = \sum_{\vec x\in A}
\rho_{s}(\vec x)$ for any countable set $A$. For any vector $\vec c \in \R^n$, we define $\rho_{s, \vec c}(\vec x) :=
\rho_s(\vec x - \vec c)$ to be a shifted version of $\rho_s$.
The following simple claim bounds the amount by which $\rho_s(\vec x)$ can shrink by a small change in $\vec x$.

\begin{claim}\label{clm:gaussianperturb}
For all $s,t,l>0$ and $\vec x, \vec y \in \R^n$ with $\|\vec x\| \le t$ and $\|\vec x-\vec y\| \le l$,
$$ \rho_s(\vec y) \ge (1- \pi (2lt+l^2)/s^2) \rho_s(\vec x).$$
\end{claim}
\begin{proof}
Using the inequality $e^{-z} \ge 1-z$,
\begin{align*}
\rho_s(\vec y) = e^{-\pi \|\vec y/s\|^2} &\ge e^{-\pi (\|\vec x\|/s + l/s)^2}
= e^{-\pi (2 l \|\vec x\|/s^2 + (l/s)^2)} \rho_s(\vec x)
\ge (1 - \pi (2 l t + l^2)/s^2) \rho_s(\vec x).
\end{align*}
\end{proof}

For any countable set $A$ and a parameter $s>0$, we define the \emph{discrete Gaussian probability distribution} $D_{A,s}$ as
\begin{equation}\label{eq:discrete_gaussian}
 \forall \vec x \in A, ~ D_{A,s}(\vec x) :=
   \frac{\rho_{s}(\vec x)}{\rho_{s}(A)}.
\end{equation}
See Figure~\ref{fig:dist_d} for an illustration.

For $\beta\in \R^+$ the distribution $\Psi_\beta$ is the distribution on $\T$ obtained by sampling from a normal
variable with mean $0$ and standard deviation $\frac{\beta}{\sqrt{2\pi}}$ and reducing the result modulo $1$ (i.e., a
periodization of the normal distribution),
\begin{align}\label{def:Psi}
 \forall r \in [0,1), ~ \Psi_\beta(r) :=  \sum_{k=-\infty}^{\infty} \frac{1}{\beta} \cdot \expo{-\pi\Big(\frac{r-k}{\beta}\Big)^2}.
\end{align}
Clearly, one can efficiently sample from $\Psi_\beta$. The following technical claim shows that a small
change in the parameter $\beta$ does not change the distribution $\Psi_\beta$ by much.

\begin{claim}\label{clm:stat_dist_norm}
For any $0<\alpha<\beta \le 2\alpha$,
 $$ \Delta(\Psi_\alpha, \Psi_\beta) \le 9 \Big(\frac{\beta}{\alpha} - 1\Big).$$
\end{claim}
\begin{proof}
We will show that the statistical distance between a normal variable with standard deviation
$\alpha/\sqrt{2\pi}$ and one with standard deviation $\beta/\sqrt{2\pi}$ is at most
$9 (\frac{\beta}{\alpha} - 1)$. This implies the claim since applying a function (modulo $1$
in this case) cannot increase the statistical distance.
By scaling, we can assume without loss of generality that $\alpha=1$ and $\beta=1+\eps$ for some $0<\eps\le 1$.
Then the statistical distance that we wish to bound is given by
\begin{align*}
 \int_\R \bigg| e^{-\pi x^2} - \frac{1}{1+\eps} \, e^{-\pi x^2 / (1+\eps)^2} \bigg| dx
  &\le
  \int_\R \big| e^{-\pi x^2} - e^{-\pi x^2 / (1+\eps)^2} \big| dx
  + \int_\R \bigg| \Big( 1 - \frac{1}{1+\eps} \Big) e^{-\pi x^2 / (1+\eps)^2} \bigg| dx \\
  &=
  \int_\R \big| e^{-\pi x^2} - e^{-\pi x^2 / (1+\eps)^2} \big| dx
  + \eps \\
  &=
  \int_\R \big| e^{-\pi ( 1 - 1/(1+\eps)^2) x^2} - 1  \big| \, e^{-\pi x^2 / (1+\eps)^2}dx
  + \eps.
\end{align*}
Now, since $1-z \le e^{-z} \le 1$ for all $z \ge 0$,
\begin{align*}
\big| e^{-\pi ( 1 - 1/(1+\eps)^2) x^2} - 1  \big| \le
\pi ( 1 - 1/(1+\eps)^2) x^2 \le 2 \pi \eps x^2.
\end{align*}
Hence we can bound the statistical distance above by
\begin{align*}
\eps + 2 \pi  \eps \int_\R x^2 e^{-\pi x^2 / (1+\eps)^2}  dx =
\eps + \eps(1+\eps)^3 \le 9\eps.
\end{align*}
\end{proof}

For an arbitrary probability distribution with density function
$\phi:\T \rightarrow \R^+$ and some integer $p \ge 1$ we define its discretization $\bar{\phi}:\Z_p \rightarrow \R^+$ as
the discrete probability distribution obtained by sampling from $\phi$, multiplying by $p$, and rounding to the closest
integer modulo $p$. More formally,
\begin{align}\label{def:discretiz}
 \bar{\phi}(i) := \int_{(i-1/2)/p}^{(i+1/2)/p} \phi(x) dx.
\end{align}
As an example, $\bar{\Psi}_\beta$ is shown in Figure~\ref{fig:discretized}.

Let $p \ge 2$ be some integer, and let $\chi:\Z_p \rightarrow \R^+$ be some probability distribution on $\Z_p$. Let $n$
be an integer and let $\vec{s} \in \Z_p^n$ be a vector. We define $A_{\vec{s},\chi}$ as the distribution on $\Z_p^n
\times \Z_p$ obtained by choosing a vector $\vec{a} \in \Z_p^n$ uniformly at random, choosing $e\in \Z_p$ according to
$\chi$, and outputting $(\vec{a},\ip{\vec{a},\vec{s}}+e)$, where additions are performed in $\Z_p$, i.e., modulo $p$.
We also define $U$ as the uniform distribution on $\Z_p^n \times \Z_p$.

For a probability density function $\phi$ on $\T$, we define $A_{\vec{s},\phi}$ as the distribution on $\Z_p^n \times
\T$ obtained by choosing a vector $\vec{a} \in \Z_p^n$ uniformly at random, choosing $e\in \T$ according to $\phi$, and
outputting $(\vec{a},\ip{\vec{a},\vec{s}}/p+e)$, where the addition is performed in $\T$, i.e., modulo $1$.

\paragraph{Learning with errors.}
For an integer $p=p(n)$ and a distribution $\chi$ on $\Z_p$, we say that an algorithm solves $\LWE_{p,\chi}$ if,
for any $\vec s \in \Z_p^n$, given samples from $A_{\vec{s},\chi}$ it outputs $\vec s$ with probability
exponentially close to $1$. Similarly, for a probability density function $\phi$ on $\T$, we say that an algorithm solves $\LWE_{p,\phi}$
if, for any $\vec s \in \Z_p^n$, given samples from $A_{\vec{s},\phi}$ it outputs $\vec s$ with probability
exponentially close to $1$. In both cases, we say that the algorithm is efficient if it runs in polynomial time in $n$.
Finally, we note that $p$ is assumed to be prime only in Lemma~\ref{lem:decision_search}; In the rest of the paper,
including the main theorem, $p$ can be an arbitrary integer.

\paragraph{Lattices.}
We briefly review some basic definitions; for a good introduction to lattices, see \cite{MicciancioBook}. A {\em lattice} in
$\R^n$ is defined as the set of all integer combinations of $n$ linearly independent vectors. This set of vectors is
known as a {\em basis} of the lattice and is not unique. Given a basis $(\vec v_1,\ldots,\vec v_n)$ of a lattice $L$, the
{\em fundamental parallelepiped} generated by this basis is defined as
$$\calP(\vec v_1,\ldots,\vec v_n) = \left\{ \left. \sum_{i=1}^n x_i \vec v_i ~\right|~ x_i \in [0,1) \right\}.$$
When the choice of basis is clear, we write $\calP(L)$ instead of $\calP(\vec v_1,\ldots,\vec v_n)$.
For a point $\vec x\in \R^n$ we define $\vec x ~\mod~ \calP(L)$ as the unique
point $\vec y\in \calP(L)$ such that $\vec y-\vec x \in L$.
We denote by $\det(L)$ the volume of the
fundamental parallelepiped of $L$ or equivalently, the absolute value of the determinant of the matrix whose columns are the basis vectors of
the lattice ($\det(L)$ is a lattice invariant, i.e., it is independent of the choice of basis). The {\em dual} of a lattice
$L$ in $\R^n$, denoted $L^*$, is the lattice given by the set of all vectors $\vec y\in \R^n$ such that $\ip{\vec x,\vec y}\in \Z$ for all vectors $\vec x\in
L$. Similarly, given a basis $(\vec v_1,\ldots,\vec v_n)$ of a lattice, we define the dual basis as the set of vectors
$(\vec v_1^*,\ldots,\vec v_n^*)$ such that $\ip{\vec v_i,\vec v_j^*} = \delta_{ij}$ for all $i,j\in [n]$ where
$\delta_{ij}$ denotes the Kronecker delta, i.e., $1$ if $i=j$ and $0$ otherwise. With a slight abuse of notation, we
sometimes write $L$ for the $n \times n$ matrix whose columns are $\vec v_1,\ldots,\vec v_n$. With this notation, we
notice that $L^* = (L^T)^{-1}$. From this it follows that $\det(L^*)=1/\det(L)$. As another example of this notation,
for a point $\vec{v} \in L$ we write $L^{-1}\vec{v}$ to indicate the integer coefficient vector of $\vec{v}$.

Let $\lambda_1(L)$ denote the length of the shortest nonzero vector in the lattice $L$. We denote by $\lambda_n(L)$ the
minimum length of a set of $n$ linearly independent vectors from $L$, where the length of a set is defined as the
length of longest vector in it. For a lattice $L$ and a point $\vec v$ whose distance from $L$ is less than $\lambda_1(L)/2$
we define $\kappa_L(\vec v)$ as the (unique) closest point to $\vec v$ in $L$. The following useful fact,
due to Banaszczyk, is known as a `transference theorem'. We remark that the lower bound is easy to prove.

\begin{lemma}[\cite{Banaszczyk}, Theorem 2.1]\label{lem:transference}
For any lattice $n$-dimensional $L$, $1 \le \lambda_1(L) \cdot \lambda_n(L^*) \le n$.
\end{lemma}

\noindent
Two other useful facts by Banaszczyk are the following. The first bounds the amount by
which the Gaussian measure of a lattice changes by scaling; the second shows that
for any lattice $L$, the mass given by the discrete Gaussian measure $D_{L,r}$
to points of norm greater than $\sqrt{n}r$ is at most exponentially small (the analogous statement for the continuous
Gaussian $\nu_r$ is easy to establish).

\begin{lemma}[\cite{Banaszczyk}, Lemma 1.4(i)]\label{lem:bana1}
For any lattice $L$ and $a \ge 1$, $\rho_a(L) \le a^n \rho(L)$.
\end{lemma}

\begin{lemma}[\cite{Banaszczyk}, Lemma 1.5(i)]\label{lem:bana2}
Let $B_n$ denote the Euclidean unit ball. Then, for any lattice $L$ and any $r>0$,
 $ \rho_r(L \setminus \sqrt{n} r B_n) < 2^{-2n} \cdot \rho_r(L)$,
 where $L \setminus \sqrt{n} r B_n$ is the set of lattice points
 of norm greater than $\sqrt{n} r$.
\end{lemma}

In this paper we consider the following lattice problems. The first two, the decision
version of the shortest vector problem ($\GapSVP$) and
the shortest independent vectors problem ($\SIVP$), are among
the most well-known lattice problems and are concerned with $\lambda_1$ and $\lambda_n$, respectively.
In the definitions below, $\gamma = \gamma(n) \ge 1$ is the approximation factor, and
the input lattice is given in the form of some arbitrary basis.

\begin{definition}
An instance of $\GapSVP_\gamma$ is given by an $n$-dimensional lattice $L$ and a number $d>0$.
In $\YES$ instances, $\lambda_1(L)\leq d$ whereas in $\NO$ instances
$\lambda_1(L)> \gamma(n) \cdot d$.
\end{definition}

\begin{definition}
An instance of $\SIVP_\gamma$ is given by an $n$-dimensional lattice $L$.
The goal is to output a set of $n$ linearly independent
lattice vectors of length at most $\gamma(n) \cdot \lambda_n(L)$.
\end{definition}

\noindent
A useful generalization of $\SIVP$ is the following somewhat less standard
lattice problem, known as the generalized independent vectors problem
($\GIVP$). Here, $\varphi$ denotes an arbitrary real-valued function on lattices. Choosing
$\varphi=\lambda_n$ results in $\SIVP$.

\begin{definition}
An instance of $\GIVP^\varphi_{\gamma}$ is given by an $n$-dimensional lattice $L$.
The goal is to output a set of $n$ linearly independent lattice
vectors of length at most $\gamma(n) \cdot \varphi(L)$.
\end{definition}

\noindent
Another useful (and even less standard) lattice problem is the following. We call it
the discrete Gaussian sampling problem ($\DGS$). As before, $\varphi$ denotes
an arbitrary real-valued function on lattices.

\begin{definition}
An instance of $\DGS_\varphi$ is given by an $n$-dimensional lattice $L$
and a number $r > \varphi(L)$. The goal is to output a sample
from $D_{L,r}$.
\end{definition}

\noindent
We also consider a variant of the closest
vector problem (which is essentially what is known as the bounded distance decoding problem \cite{LiLyMi06}): For an $n$-dimensional lattice $L$, and some $d>0$, we say that an algorithm solves $\CVP_{L,d}$ if, given a
point $\vec x \in \R^n$ whose distance to $L$ is at most $d$, the algorithm finds the closest lattice point to $\vec
x$. In this paper $d$ will always be smaller than $\lambda_1(L)/2$ and hence the closest vector is
unique.

\paragraph{The smoothing parameter.}
We make heavy use of a lattice parameter known as the {\em smoothing parameter} \cite{MicciancioR04}.
Intuitively, this parameter provides the width beyond which the discrete Gaussian measure
on a lattice behaves like a continuous one. The precise definition is the following.

\begin{definition}\label{def:smoothingpara}
For an $n$-dimensional lattice $L$ and positive real $\epsilon>0$, we define the smoothing parameter
$\eta_\epsilon(L)$ to be the smallest $s$ such that $\rho_{1/s}(L^* \setminus \{\vec 0\}) \leq \epsilon$.
\end{definition}

\noindent
In other words, $\eta_\eps(L)$ is the smallest $s$ such that a Gaussian measure scaled by $1/s$ on the dual lattice $L^*$
gives all but $\eps/(1+\eps)$ of its weight to the origin. We usually take $\eps$ to be some negligible function
of the lattice dimension $n$. Notice that $\rho_{1/s}(L^* \setminus\{\vec 0\})$ is a continuous and strictly decreasing
function of $s$. Moreover, it can be shown that $\lim_{s\to 0}\rho_{1/s}(L^* \setminus\{\vec 0\}) = \infty$ and $\lim_{s\to
\infty}\rho_{1/s}(L^* \setminus\{\vec 0\}) = 0$. So, the parameter $\eta_\epsilon(L)$ is well defined for any
$\epsilon>0$, and $\epsilon \mapsto \eta_\epsilon(L)$ is the inverse function of $s \mapsto \rho_{1/s}(L^* \setminus \{
\vec 0\})$. In particular, $\eta_{\epsilon}(L)$ is also a continuous and strictly decreasing function of $\epsilon$.

The motivation for this definition (and the name `smoothing parameter') comes from the following result, shown in
\cite{MicciancioR04} (and included here as Claim~\ref{clm:gaussianweightshiftinvariant}). Informally, it says that if we choose a `random' lattice point from an $n$-dimensional lattice
$L$ and add continuous Gaussian noise $\nu_s$ for some $s > \eta_\eps(L)$ then the resulting distribution is within statistical distance
$\eps$ of the `uniform distribution on $\R^n$'. In this
paper, we show (in Claim~\ref{clm:gaussian_noise}) another important property of this parameter: for $s > \sqrt{2} \eta_\eps(L)$, if we sample a point from
$D_{L,s}$ and add Gaussian noise $\nu_s$, we obtain a distribution whose statistical distance to a continuous Gaussian
$\nu_{\sqrt{2} s}$ is at most $4\eps$. Notice that $\nu_{\sqrt{2} s}$ is the distribution one obtains when
summing two independent samples from $\nu_s$. Hence, intuitively, the noise $\nu_s$ is enough to hide the
discrete structure of $D_{L,s}$.

The following two upper bounds on the smoothing parameter appear in \cite{MicciancioR04}.

\begin{lemma}\label{lem:lambda1}
For any $n$-dimensional lattice $L$, $\eta_\epsilon(L) \leq \sqrt{n}/\lambda_1(L^*)$ where $\epsilon=2^{-n}$.
\end{lemma}

\begin{lemma}\label{lem:lambdan}
For any $n$-dimensional lattice $L$ and $\epsilon>0$,
\[ \eta_\epsilon(L) \leq \sqrt{\frac{\ln(2n(1+1/\epsilon))}{\pi}} \cdot \lambda_n(L).\]
In particular, for any superlogarithmic function $\omega(\log n)$, $\eta_{\epsilon(n)}(L) \leq \sqrt{\omega(\log n)}
\cdot \lambda_n(L)$ for some negligible function $\epsilon(n)$.
\end{lemma}

We also need the following simple lower bound on the smoothing parameter.

\begin{claim}\label{clm:lowerboundsmoothing}
For any lattice $L$ and any $\eps>0$,
$$ \eta_\eps(L) \ge \sqrt{\frac{\ln 1/\eps}{\pi}} \cdot \frac{1}{\lambda_1(L^*)}
 \ge \sqrt{\frac{\ln 1/\eps}{\pi}} \cdot \frac{\lambda_n(L)}{n}.$$
In particular, for any $\eps(n)=o(1)$ and any constant $c>0$, $\eta_{\eps(n)}(L) > c / \lambda_1(L^*) \ge c \lambda_n(L) / n$
for large enough $n$.
\end{claim}
\begin{proof}
Let $\vec v\in L^*$ be a vector of length $\lambda_1(L^*)$ and let $s=\eta_\eps(L)$. Then,
$$ \eps = \rho_{1/s}(L^* \setminus \{\vec 0\}) \ge \rho_{1/s}(\vec v) = \expo{-\pi (s \lambda_1(L^*))^2}.$$
The first inequality follows by solving for $s$. The second inequality is by Lemma~\ref{lem:transference}.
\end{proof}

\paragraph{The Fourier transform.}
We briefly review some of the important properties of the Fourier transform. In the following, we omit certain
technical conditions as these will always be satisfied in our applications. For a more precise and in-depth treatment,
see, e.g., \cite{LatticeAndCodesBook}. The Fourier transform of a function $h:\R^n\to \C$ is defined to be
$$\hat{h}(\vec w)=\int_{\R^n} h(\vec x) e^{-2\pi i\ip{\vec x,\vec w}} d\vec x.$$
From the definition we can obtain two useful formulas; first, if $h$ is defined by $h(\vec x)=g(\vec x+\vec v)$ for
some function $g$ and vector $\vec v$ then
\begin{equation}\label{eq:fourier_shift}
\hat{h}(\vec w) = e^{2\pi i \ip{\vec v, \vec w}} \hat{g}(\vec w).
\end{equation}
Similarly, if $h$ is defined by $h(\vec x) = e^{2\pi i \ip{\vec x, \vec v}} g(\vec x)$ for some function $g$ and vector
$\vec v$ then
\begin{equation}\label{eq:fourier_shift_2}
\hat{h}(\vec w) = \hat{g}(\vec w - \vec v).
\end{equation}

Another important fact is that the Gaussian is its own Fourier transform, i.e., $\hat{\rho}=\rho$. More generally, for
any $s>0$ it holds that $\widehat{\rho_s}=s^n \rho_{1/s}$.
Finally, we will use the following formulation of the Poisson summation formula.
\begin{lemma}[Poisson summation formula]\label{lem:psf}
For any lattice $L$ and any function $f: \R^n \to \C$,
$$ f(L) = \det(L^*) \hat{f}(L^*).$$
\end{lemma}

\section{Main Theorem}\label{sce:main_theorem}

Our main theorem is the following. The connection to the standard lattice problems $\GapSVP$
and $\SIVP$ will be established in Subsection~\ref{ssec:standardlatticeproblems} by
polynomial time reductions to $\DGS$.

\begin{theorem}[Main theorem]\label{thm:mainthmmm}
Let $\eps=\eps(n)$ be some negligible function of $n$. Also, let $p=p(n)$ be some integer
and $\alpha = \alpha (n) \in (0,1)$ be such that $\alpha p
> 2\sqrt{n}$.
Assume that we have access to an oracle $W$ that solves $\LWE_{p,\Psi_\alpha}$ given a polynomial number of samples.
Then there exists an efficient quantum algorithm for $\DGS_{\sqrt{2n} \cdot \eta_\eps(L) /\alpha}$.
\end{theorem}
\begin{proof}
The input to our algorithm is an $n$-dimensional lattice $L$
and a number $r > \sqrt{2n} \cdot \eta_\eps(L) /\alpha$. Our goal is to output a sample
from $D_{L,r}$.
Let $r_i$ denote $r \cdot (\alpha p / \sqrt{n})^i$. The algorithm starts by producing
$n^{c}$ samples from $D_{L,r_{3n}}$ where $c$ is the constant from the iterative step lemma, Lemma~\ref{lem:iterative_step}.
By Claim~\ref{clm:lowerboundsmoothing}, $r_{3n} > 2^{3n} r > 2^{2n} \lambda_n(L)$, and hence
we can produce these samples efficiently by the procedure described in the bootstrapping lemma,
Lemma~\ref{lem:bootstrap}. Next, for $i = 3n, 3n-1,\ldots,1$ we use our $n^{c}$
samples from $D_{L,r_i}$ to produce $n^{c}$ samples from $D_{L,r_{i-1}}$.
The procedure that does this, called the iterative step, is the core of the algorithm and is described
in Lemma~\ref{lem:iterative_step}.
Notice that the condition in Lemma~\ref{lem:iterative_step} is satisfied since for all $i \ge 1$,
$r_i \ge r_1 = r \alpha p / \sqrt{n} > \sqrt{2} p \eta_\eps(L)$. At the end of the loop, we end
up with $n^{c}$ samples from $D_{L,r_0} = D_{L,r}$ and we complete the algorithm
by simply outputting the first of those.
\end{proof}

\subsection{Bootstrapping}

\begin{lemma}[Bootstrapping]\label{lem:bootstrap}
There exists an efficient algorithm that, given any $n$-dimensional lattice $L$ and $r > 2^{2n} \lambda_n(L)$,
outputs a sample from a distribution that is within statistical distance $2^{-\Omega(n)}$ of $D_{L,r}$.
\end{lemma}
\begin{proof}
By using the LLL basis reduction algorithm \cite{LLL}, we obtain a basis for $L$ of length at most $2^n
\lambda_n(L)$ and let $\calP(L)$ be the parallelepiped generated by this basis.
The sampling procedure samples a vector $\vec{y}$ from $\nu_r$ and then outputs $\vec y -
(\vec y ~\mod~ \calP(L)) \in L$. Notice that $\|\vec y ~\mod~ \calP(L)\| \le \diam(\calP(L)) \le n 2^n \lambda_n(L)$.

Our goal is to show that the resulting distribution is exponentially close to $D_{L,r}$.
By Lemma~\ref{lem:bana2}, all but an exponentially small part of $D_{L,r}$
is concentrated on points of norm at most $\sqrt{n} r$. So consider any $\vec x \in L$
with $\| \vec x\| \le \sqrt{n} r$. By definition, the probability given to it by $D_{L,r}$ is
$\rho_r(\vec x) / \rho_r(L)$. By Lemma~\ref{lem:psf}, the denominator is
$\rho_r(L) = \det(L^*) \cdot r^n \rho_{1/r}(L^*) \ge \det(L^*) \cdot r^n$ and hence the probability
is at most $\rho_r(\vec x) / (\det(L^*) \cdot r^n) = \det(L) \nu_r(\vec x)$.
On the other hand, by Claim~\ref{clm:gaussianperturb}, the probability given to $\vec x \in L$ by our procedure is
$$ \int_{\vec x + \calP(L)} \nu_r(\vec y) d \vec y \ge (1-2^{-\Omega(n)})\det(L) \nu_r(\vec x).$$
Together, these facts imply that our output distribution is within statistical distance $2^{-\Omega(n)}$
of $D_{L,r}$.
\end{proof}

\subsection{The iterative step}

\begin{lemma}[The iterative step]\label{lem:iterative_step}
Let $\eps=\eps(n)$ be a negligible function, $\alpha = \alpha(n) \in (0,1)$ be
a real number, and $p =p(n) \ge 2$ be an integer.
Assume that we have access to an oracle $W$ that solves $\LWE_{p,\Psi_\alpha}$ given a polynomial number of samples.
Then, there exists a constant $c>0$ and an efficient quantum algorithm that, given any
$n$-dimensional lattice $L$, a number $r>\sqrt{2} p
\eta_\eps(L)$, and $n^{c}$ samples from $D_{L,r}$,
produces a sample from $D_{L,r\sqrt{n}/(\alpha p)}$.
\end{lemma}

\noindent
Note that the output distribution is taken with respect to the randomness (and quantum measurements) used in the
algorithm, and not with respect to the input samples. In particular, this means that from the same set
of $n^{c}$ samples from $D_{L,r}$ we can produce any polynomial number of samples from $D_{L,r\sqrt{n}/(\alpha p)}$.

\begin{proof}
The algorithm consists of two main parts. The first part is shown in Lemma~\ref{lem:from_samples_to_cvp}. There, we
describe a (classical) algorithm that using $W$ and the samples from $D_{L,r}$, solves $\CVP_{L^*, \alpha p / (\sqrt{2}r)}$. The
second part is shown in Lemma~\ref{lem:quantum_part}. There, we describe a quantum algorithm that, given an oracle
that solves $\CVP_{L^*, \alpha p / (\sqrt{2}r)}$, outputs a sample from $D_{L,r\sqrt{n}/(\alpha p)}$. This is the only quantum
component in this paper.
We note that the condition in Lemma~\ref{lem:quantum_part} is satisfied since by Claim~\ref{clm:lowerboundsmoothing},
$ \alpha p/(\sqrt{2} r) \le 1/\eta_\eps(L) \le \lambda_1(L^*)/2$.
\end{proof}

\subsubsection{From samples to $\CVP$}

Our goal in this subsection is to prove the following.

\begin{lemma}[First part of iterative step]\label{lem:from_samples_to_cvp}
Let $\eps=\eps(n)$ be a negligible function, $p =p(n) \ge 2$ be an integer, and
$\alpha = \alpha(n) \in (0,1)$ be a real number.
Assume that
we have access to an oracle $W$ that solves $\LWE_{p, \Psi_\alpha}$ given a polynomial number of samples.
Then, there exist a constant $c>0$ and an efficient algorithm that, given any $n$-dimensional lattice $L$, a number $r> \sqrt{2} p \eta_\eps(L)$,
and $n^c$ samples from $D_{L,r}$, solves $\CVP_{L^*, \alpha p/(\sqrt{2}r)}$.
\end{lemma}

For an $n$-dimensional lattice $L$, some $0<d<\lambda_1(L)/2$, and an integer $p \ge 2$,
we say that an algorithm solves $\CVP_{L,d}^{(p)}$ if, given any point $\vec x \in \R^n$ within distance $d$ of $L$,
it outputs $L^{-1}\kappa_L(\vec x) ~\mod~ p \in \Z_p^n$, the coefficient vector of the closest vector to $\vec{x}$
reduced modulo $p$. We start with the following lemma, which shows a reduction from $\CVP_{L,d}$ to $\CVP_{L,d}^{(p)}$.

\begin{lemma}[Finding coefficients modulo $p$ is sufficient]\label{lem:cvpmodp}
There exists an efficient algorithm that given a lattice $L$, a number $d < \lambda_1(L)/2$
and an integer $p \ge 2$, solves $\CVP_{L,d}$ given access to an oracle for $\CVP_{L,d}^{(p)}$.
\end{lemma}
\begin{proof}
Our input is a point $\vec x$ within distance $d$ of $L$.
We define a sequence of points $\vec x_1=\vec x,\vec x_2,\vec x_3,\ldots$ as follows. Let $\vec a_i=L^{-1}\kappa_L(\vec
x_i) \in \Z^n$ be the coefficient vector of the closest lattice point to $\vec x_i$. We define $\vec x_{i+1}=(\vec x_i
- L(\vec a_i ~\mod ~ p))/p$. Notice that the closest lattice point to $\vec x_{i+1}$ is $L(\vec a_i-(\vec a_i ~\mod~
p))/p \in L$ and hence $\vec a_{i+1}= (\vec a_i-(\vec a_i ~\mod~ p))/p$. Moreover, the distance of $\vec x_{i+1}$ from
$L$ is at most $d/p^i$. Also note that this sequence can be computed by using the oracle.

After $n$ steps, we have a point $\vec x_{n+1}$ whose distance to the lattice is at most $d/p^n$. We now apply an
algorithm for approximately solving the closest vector problem, such as Babai's nearest plane algorithm~\cite{Babai86}.
This yields a lattice point $L\vec a$ within distance $2^n \cdot d/p^n \le d < \lambda_1(L)/2$ of $\vec x_{n+1}$.
Hence, $L\vec a$ is the lattice point closest to $\vec x_{n+1}$ and we managed to recover $\vec a_{n+1}=\vec a$.
Knowing $\vec a_{n+1}$ and $\vec a_n ~\mod~ p$ (by using the oracle), we can now recover $\vec a_n = p \vec a_{n+1} +
(\vec a_n ~\mod~ p)$. Continuing this process, we can recover $\vec a_{n-1},\vec a_{n-2},\ldots,\vec a_1$. This
completes the algorithm since $L\vec a_1$ is the closest point to $\vec x_1=\vec x$.
\end{proof}

As we noted in the proof of Lemma~\ref{lem:iterative_step}, for our choice of $r$,
$ \alpha p/(\sqrt{2} r) \le \lambda_1(L^*)/2$.
Hence, in order to prove Lemma~\ref{lem:from_samples_to_cvp},
it suffices to present an efficient algorithm for $\CVP_{L^*, \alpha p/(\sqrt{2}r)}^{(p)}$.
We do this by combining two lemmas.
The first, Lemma~\ref{lem:learning_smooth}, shows an algorithm $W'$ that, given samples from $A_{\vec{s},\Psi_\beta}$
for some (unknown) $\beta \le \alpha$, outputs $\vec{s}$ with probability exponentially close to $1$ by using
$W$ as an oracle.
Its proof is based on Lemma~\ref{lem:verification}.
The second, Lemma~\ref{lem:find_coeff}, is the main lemma of this subsection, and shows how to use $W'$ and the given
samples from $D_{L,r}$ in order to solve $\CVP_{L^*, \alpha p/(\sqrt{2}r)}^{(p)}$.

\begin{lemma}[Verifying solutions of $\LWE$]\label{lem:verification}
Let $p=p(n) \ge 1$ be some integer. There exists an efficient algorithm that, given $\vec s'$ and samples from $A_{\vec{s},\Psi_\alpha}$ for some (unknown)
$\vec{s} \in \Z_p^n$ and $\alpha < 1$, outputs whether $\vec s = \vec s'$ and is correct with probability exponentially close to
$1$.
\end{lemma}

\noindent We remark that the lemma holds also for all $\alpha \le O(\sqrt{\log n})$ with essentially the same proof.

\begin{proof}
The idea is to perform a statistical test on samples from $A_{\vec{s},\Psi_\alpha}$ that checks whether $\vec s = \vec
s'$. Let $\xi$ be the distribution on $\T$ obtained by sampling $(\vec a, x)$ from $A_{\vec{s},\Psi_\alpha}$ and
outputting $x - \ip{\vec a, \vec s'}/p \in \T$. The algorithm takes $n$ samples
$y_1,\ldots,y_n$ from $\xi$. It then computes $z := \frac{1}{n} \sum_{i=1}^n \cos(2 \pi y_i)$. If $z > 0.02$, it
decides that $\vec s = \vec s'$, otherwise it decides that $\vec s \neq \vec s'$.

We now analyze this algorithm. Consider the distribution $\xi$. Notice that it be obtained by sampling $e$ from
$\Psi_\alpha$, sampling $\vec a$ uniformly from $\Z_p^n$ and outputting $e + \ip{\vec a, \vec s - \vec s'}/p \in \T$.
From this it easily follows that if $\vec s = \vec s'$, $\xi$ is exactly $\Psi_\alpha$. Otherwise, if $\vec s \neq \vec
s'$, we claim that $\xi$ has a period of $1/k$ for some integer $k\ge 2$. Indeed, let $j$ be an index on which $s_j
\neq s'_j$. Then the distribution of $a_j (s_j - s'_j) ~\mod~ p$ is periodic with period ${\rm{gcd}}(p, s_j - s'_j) < p$.
This clearly implies that the distribution of $a_j (s_j - s'_j)/p ~\mod~ 1$ is periodic with period $1/k$ for some $k
\ge 2$. Since a sample from $\xi$ can be obtained by adding a sample from $a_j (s_j - s'_j)/p ~\mod~ 1$ and an
independent sample from some other distribution, we obtain that $\xi$ also has the same period of $1/k$.

Consider the expectation\footnote{We remark that this expectation is essentially the Fourier series of $\xi$
at point $1$ and that the following arguments can be explained in terms of properties of the Fourier series.}
\begin{align*}
 \tilde{z} := \Exp_{y \sim \xi} [ \cos(2 \pi y) ] = \int_0^1 \cos(2\pi y) \xi(y) dy =
 \Re \Big[ \int_0^1 \expo{2\pi i y} \xi(y) dy\Big].
\end{align*}
First, a routine calculation shows that for $\xi = \Psi_\alpha$, $\tilde{z} = \expo{- \pi
\alpha^2}$, which is at least $0.04$ for $\alpha < 1$. Moreover, if $\xi$ has a period of $1/k$, then
$$ \int_0^1 \expo{2\pi i y} \xi(y) dy = \int_0^1 \expo{2\pi i (y+\textstyle{\frac{1}{k}})} \xi(y) dy
  = \expo{2\pi i /k }  \int_0^1 \expo{2\pi i y} \xi(y) dy$$
which implies that if $k \ge 2$ then $\tilde{z}=0$.
We complete the proof by noting that by the Chernoff bound, $|z-\tilde{z}| \le 0.01$
with probability exponentially close to $1$.
\end{proof}

\begin{lemma}[Handling error $\Psi_\beta$ for $\beta \le \alpha$]\label{lem:learning_smooth}
Let $p =p(n) \ge 2$ be some integer and $\alpha=\alpha(n) \in (0,1)$.
Assume that we have access to an oracle $W$ that solves $\LWE_{p, \Psi_\alpha}$ by using a polynomial
number of samples. Then, there
exists an efficient algorithm $W'$ that, given samples from $A_{\vec{s},\Psi_\beta}$ for some (unknown) $\beta \le \alpha$, outputs
$\vec{s}$ with probability exponentially close to $1$.
\end{lemma}
\begin{proof}
The proof is based on the following idea: by adding the right amount of noise, we can transform samples from
$A_{\vec{s},\Psi_\beta}$ to samples from $A_{\vec{s},\Psi_\alpha}$ (or something sufficiently close to it). Assume that the number of samples required by $W$
is at most $n^c$ for some $c>0$. Let $Z$ be the set of all integer multiplies of $n^{-2c} \alpha^2$ between $0$ and
$\alpha^2$. For each $\gamma \in Z$, Algorithm $W'$ does the following $n$ times. It takes $n^c$ samples from
$A_{\vec{s},\Psi_\beta}$ and adds to the second element of each sample a noise sampled independently from
$\Psi_{\sqrt{\gamma}}$. This creates $n^c$ samples taken from the distribution $A_{\vec{s},\Psi_{\sqrt{\beta^2+\gamma}}}$. It
then applies $W$ and obtains some candidate $\vec{s}'$. Using Lemma~\ref{lem:verification}, it checks whether
$\vec{s}'=\vec{s}$. If the answer is yes, it outputs $\vec s'$;
otherwise, it continues.

We now show that $W'$ finds $\vec s$ with probability exponentially close to $1$. By Lemma~\ref{lem:verification}, if
$W'$ outputs some value, then this value is correct with probability exponentially close to $1$. Hence, it is enough to
show that in one of the iterations, $W'$ outputs some value. Consider the smallest $\gamma \in Z$ such that $\gamma \ge
\alpha^2 - \beta^2$. Clearly, $\gamma \le \alpha^2-\beta^2 + n^{-2c}\alpha^2$. Define $\alpha' =
\sqrt{\beta^2+\gamma}$. Then,
$$\alpha \le \alpha' \le \sqrt{\alpha^2 + n^{-2c} \alpha^2} \le (1+n^{-2c}) \alpha.$$
By Claim~\ref{clm:stat_dist_norm}, the statistical distance between $\Psi_\alpha$ and $\Psi_{\alpha'}$ is at most
$9n^{-2c}$. Hence, the statistical distance between $n^c$ samples from $\Psi_\alpha$ and $n^c$ samples from
$\Psi_{\alpha'}$ is at most $9n^{-c}$. Therefore, for our choice of $\gamma$, $W$ outputs $\vec s$ with probability at
least $1-9n^{-c}/2-2^{-\Omega(n)} \ge \frac{1}{2}$. The probability that $\vec s$ is not found in any of the $n$ calls to
$W$ is at most $2^{-n}$.
\end{proof}

For the analysis of our main procedure in Lemma~\ref{lem:find_coeff},
we will need to following claims regarding Gaussian
measures on lattices.
On first reading, the reader can just read the statements
of Claim~\ref{clm:gaussianweightshiftinvariant} and Corollary~\ref{cor:gaussian_noise_one_dim}
and skip directly to Lemma~\ref{lem:find_coeff}.
All claims show that in some sense,
when working above the smoothing parameter, the discrete
Gaussian measure behaves like the continuous Gaussian
measure. We start with the following claim, showing that above
the smoothing parameter, the discrete Gaussian measure
is essentially invariant under shifts.

\begin{claim}\label{clm:gaussianweightshiftinvariant}
For any lattice $L$, $\vec c \in \R^n$, $\eps>0$, and $r \ge \eta_\eps(L)$,
$$ \rho_r(L+\vec c) \in r^n \det(L^*) (1\pm \eps).$$
\end{claim}
\begin{proof}
Using the Poisson summation formula (Lemma~\ref{lem:psf})
and the assumption that $\rho_{1/r}(L^* \setminus \{\vec 0\}) \le \eps$,
\begin{align*}
 \rho_r(L+\vec c)
 = \sum_{\vec{x} \in L} \rho_r(\vec{x}+\vec{c})
 &= \sum_{\vec x \in L} \rho_{r,-\vec c}(\vec x) \\
 &= \det(L^*)\sum_{\vec{y}\in L^*} \widehat{\rho_{r,-\vec c}}(\vec{y}) \\
 &= r^n \det(L^*)\sum_{\vec{y}\in L^*} \expo{2\pi i \ip{\vec c, \vec{y}}} \rho_{1/r}(\vec{y}) \\
 &= r^n \det(L^*)(1\pm \eps).
\end{align*}
\end{proof}

The following claim (which is only used to establish the corollary following it)
says that when adding a continuous Gaussian of width $s$ to a discrete Gaussian of width $r$,
with both $r$ and $s$ sufficiently greater than the smoothing parameter,
the resulting distribution is very close to a continuous Gaussian of the width we would expect,
namely $\sqrt{r^2+s^2}$.
To get some intuition on why we need to assume that \emph{both} Gaussians
are sufficiently wide, notice for instance that if the discrete Gaussian is very narrow, then it
is concentrated on the origin, making the sum have width $s$.
Also, if the continuous Gaussian is too narrow, then the discrete
structure is still visible in the sum.

\begin{claim}\label{clm:gaussian_noise}
Let $L$ be a lattice, let $\vec u\in \R^n$ be any vector, let $r,s > 0$ be two reals, and let $t$ denote
$\sqrt{r^2+s^2}$. Assume that $rs/t = 1/\sqrt{1/r^2+1/s^2} \ge \eta_\eps(L)$ for some $\eps < \frac12$.
Consider the
continuous distribution $Y$ on $\R^n$ obtained by sampling from $D_{L+\vec u,r}$ and then adding a noise vector taken
from $\nu_s$. Then, the statistical distance between $Y$ and $\nu_t$ is at most $4\eps$.
\end{claim}
\begin{proof}
The probability density function of $Y$ can be written as
\begin{align}
Y(\vec x) & = \frac{1}{s^n\rho_r(L+\vec u)} \sum_{\vec y \in L+\vec u}  \rho_{r}(\vec y)\rho_s(\vec x - \vec y) \nonumber \\
          & = \frac{1}{s^n\rho_r(L+\vec u)} \sum_{\vec y \in L+\vec u}  \expo{-\pi (\|\vec y/r\|^2 + \|(\vec x - \vec y)/s\|^2)} \nonumber \\
          & = \frac{1}{s^n\rho_r(L+\vec u)} \sum_{\vec y \in L+\vec u}
               \expo{-\pi \Big(  \frac{r^2+s^2}{r^2 \cdot s^2} \cdot \Big\|\vec y - \frac{r^2}{r^2+s^2} \vec x\Big\|^2 +
               \frac{1}{r^2+s^2}\|\vec x\|^2\Big )} \nonumber \\
          & = \expo{- \frac{\pi}{r^2+s^2}\|\vec x\|^2} \frac{1}{s^n\rho_r(L+\vec u)} \sum_{\vec y \in L+\vec u}
               \expo{-\pi \Big(  \frac{r^2+s^2}{r^2 \cdot s^2}  \cdot \Big\|\vec y - \frac{r^2}{r^2+s^2} \vec x\Big\|^2\Big )}
               \nonumber \\
          & = \frac{1}{s^n} \rho_{t}(\vec x) \cdot \frac{\rho_{rs/t,(r/t)^2 \vec
          x - \vec u}(L)}{\rho_{r,-\vec u}(L)} \nonumber \\
          & = \frac{1}{s^n} \rho_t (\vec x) \cdot \frac{\widehat{\rho_{rs/t, (r/t)^2 \vec
          x - \vec u}}(L^*)}{\widehat{\rho_{r, - \vec u}}(L^*)} \nonumber \\
          & = \rho_t(\vec x) / t^n
             \cdot \frac{(t/rs)^n \widehat{\rho_{rs/t,(r/t)^2 \vec x - \vec u}}(L^*)}
                  {(1/r)^n \widehat{\rho_{r, - \vec u}}(L^*)}
          \label{eq:last_in_y}
\end{align}
where in the next-to-last equality we used Lemma~\ref{lem:psf}. Using Eq.~\eqref{eq:fourier_shift},
\begin{align*}
\widehat{\rho_{rs/t,(r/t)^2 \vec x - \vec u}(\vec w)} &=
  \expo{-2\pi i \ip{(r/t)^2 \vec x - \vec u, \vec w}} \cdot  (rs/t)^n \rho_{t/rs}(\vec w), \\
\widehat{\rho_{r,- \vec u}}(\vec w) &=
  \expo{2\pi i \ip{ \vec u, \vec w}} \cdot r^n \rho_{1/r}(\vec w).
\end{align*}
Hence,
\begin{align*}
\Big|1- (t/rs)^n \widehat{\rho_{rs/t,(r/t)^2 \vec x - \vec u}}(L^*) \Big| & \le
 \rho_{t/rs}(L^* \setminus \{\vec 0\}) \le \eps \\
\Big|1- (1/r)^n \widehat{\rho_{r, - \vec u}}(L^*) \Big| & \le
 \rho_{1/r}(L^* \setminus \{\vec 0\}) \le \eps
\end{align*}
where the last inequality follows from $1/r \le t/rs$.  Hence, the quotient in \eqref{eq:last_in_y} is between
$(1-\eps)/(1+\eps) \ge 1-2\eps$ and $(1+\eps)/(1-\eps) \le 1+4\eps$. This implies that,
\begin{align*}
| Y(\vec x) - \rho_{t}(\vec x) / t^n | \le
   \rho_t(\vec x) / t^n \cdot
   4\eps.
\end{align*}
We complete the proof by integrating over $\R^n$.
\end{proof}

\begin{corollary}\label{cor:gaussian_noise_one_dim}
Let $L$ be a lattice, let $\vec z, \vec u\in \R^n$ be vectors, and let $r,\alpha > 0$ be two reals.
Assume that $1/\sqrt{1/r^2+(\|\vec z\|/\alpha)^2} \ge \eta_\eps(L)$ for some $\eps < \frac12$.
Then the distribution of
$\ip{\vec z, \vec v} + e$
where $\vec v$ is distributed according to $D_{L+\vec u, r}$ and
$e$ is a normal variable with mean $0$ and standard deviation $\alpha/\sqrt{2 \pi}$,
is within statistical distance $4 \eps$ of a normal variable with mean $0$ and standard deviation
$\sqrt{(r\|\vec z\|)^2+\alpha^2}/\sqrt{2 \pi}$.
In particular, since statistical distance cannot increase by applying a function,
the distribution of
$ \ip{\vec z, \vec v} + e ~\mod~ 1$
is within statistical distance $4 \eps$ of $\Psi_{\sqrt{(r\|\vec z\|)^2+\alpha^2}}$.
\end{corollary}
\begin{proof}
We first observe that the distribution of $\ip{\vec z, \vec v} + e$
is exactly the same as that of
$ \ip{\vec z, \vec v + \vec h} $
where $\vec h$ is distributed as the continuous Gaussian $\nu_{\alpha/\|\vec z\|}$.
Next, by Claim~\ref{clm:gaussian_noise}, we know that the distribution of $\vec v + \vec h$
is within statistical distance $4 \eps$ of the continuous Gaussian $\nu_{\sqrt{r^2+(\alpha/\|\vec z\|)^2}}$.
Taking the inner product of this continuous Gaussian with $\vec z$ leads to a normal
distribution with mean 0 and standard deviation $\sqrt{(r\|\vec z\|)^2+\alpha^2}/\sqrt{2 \pi}$,
and we complete the proof by using the fact that statistical distance cannot increase by
applying a function (inner product with $\vec z$ in this case).
\end{proof}

\begin{lemma}[Main procedure of the first part]\label{lem:find_coeff}
Let $\eps=\eps(n)$ be a negligible function, $p =p(n) \ge 2$ be an integer, and
$\alpha = \alpha(n) \in (0,1)$ be a real number.
Assume that we have access to an oracle $W$ that for all
$\beta \le \alpha$, finds $\vec{s}$ given a polynomial number of samples from $A_{\vec{s}, \Psi_\beta}$
(without knowing $\beta$).
Then, there exists an efficient algorithm that given an $n$-dimensional lattice $L$, a number
$r > \sqrt{2}p \eta_\eps(L)$, and a polynomial number of samples from $D_{L,r}$,
solves $\CVP_{L^*, \alpha p/(\sqrt{2}r)}^{(p)}$.
\end{lemma}
\begin{proof}
We describe a procedure that given $\vec{x}$ within distance $\alpha p / (\sqrt{2}r)$ of $L^*$,
outputs samples from the distribution $A_{\vec{s},\Psi_\beta}$
for some $\beta \le \alpha$ where $\vec{s} = (L^*)^{-1}\kappa_{L^*}(\vec{x}) ~\mod~ p$. By running this
procedure a polynomial number of times and then using $W$, we can find $\vec{s}$.

The procedure works as follows. We sample a vector $\vec{v} \in L$ from $D_{L,r}$, and let $\vec{a} = L^{-1}\vec{v}
~\mod~ p$. We then output
\begin{align}\label{eq:outputofprocedure}
(\vec{a},\ip{\vec{x},\vec{v}}/p+e ~\mod~1)
\end{align}
where $e \in \R$ is chosen according to a normal
distribution with standard deviation $\alpha /(2\sqrt{\pi})$.
We claim that the distribution given by this procedure is within negligible statistical distance of $A_{\vec{s},\Psi_\beta}$
for some $\beta \le \alpha$.

We first notice that the distribution of $\vec{a}$ is very close to uniform.
Indeed, the probability of obtaining each $\vec{a} \in \Z_p^n$ is proportional to
$\rho_r(pL+L\vec{a})$. Using $\eta_\epsilon(pL) = p \eta_\epsilon(L) < r$ and
Claim~\ref{clm:gaussianweightshiftinvariant},
the latter is $(r/p)^{n} \det(L^*) (1\pm \eps)$, which
implies that the statistical distance between the distribution of $\vec{a}$ and the uniform distribution is
negligible.

Next, we condition on any fixed value of $\vec a$ and
consider the distribution of the second element in \eqref{eq:outputofprocedure}. Define $\vec{x'} =
\vec{x}-\kappa_{L^*}(\vec{x})$ and note that $\|\vec{x'}\| \le \alpha p / (\sqrt{2}r)$. Then,
$$\ip{\vec{x},\vec{v}}/p+e ~\mod~1 = \ip{\vec{x'}/p,\vec{v}}+e + \ip{\kappa_{L^*}(\vec{x}),\vec{v}}/p ~\mod~1.$$
Now,
$$ \ip{\kappa_{L^*}(\vec{x}),\vec{v}} =
   \ip{(L^*)^{-1} \kappa_{L^*}(\vec{x}),L^{-1} \vec{v}}
   $$
since $L^{-1} = (L^*)^T$. In words, this says that the inner product between
$\kappa_{L^*}(\vec{x})$ and $\vec{v}$ (and in fact, between any vector in $L^*$ and any vector in $L$) is the same as the
inner product between the corresponding coefficient vectors. Since the coefficient vectors are integer,
$$ \ip{\kappa_{L^*}(\vec{x}),\vec{v}} ~\mod~ p = \ip{\vec{s},\vec{a}} ~\mod~ p$$
from which it follows that $\ip{\kappa_{L^*}(\vec{x}),\vec{v}}/p ~\mod~1$ is exactly $\ip{\vec{s},\vec{a}}/p ~\mod~ 1$.

We complete the proof by applying Corollary~\ref{cor:gaussian_noise_one_dim}, which shows that the distribution of the remaining part $\ip{\vec{x'}/p,\vec{v}}+e$
is within negligible statistical distance of $\Psi_\beta$ for $\beta=\sqrt{(r\|\vec x'\|/p)^2 + \alpha^2/2} \le \alpha$,
as required. Here we used that the distribution of $\vec{v}$ is $D_{pL+L\vec a, r}$ (since we are
conditioning on $\vec a$), the distribution of $e$ is normal with mean $0$ and standard deviation $(\alpha/\sqrt{2})/\sqrt{2 \pi}$,
and that
$$ 1/\sqrt{1/r^2+(\sqrt{2}\|\vec x'\|/p\alpha)^2} \ge r/\sqrt{2} > \eta_\eps(pL).$$
\end{proof}

\subsubsection{From $\CVP$ to samples}\label{ssec:from_cvp_to_samples}

In this subsection we describe a quantum procedure that uses a $\CVP$ oracle
in order to create samples from the discrete Gaussian distribution. We assume
familiarity with some basic notions of quantum computation, such as
(pure) states, measurements, and the quantum Fourier transform. See, e.g., \cite{nielsen&chuang:qc}
for a good introduction. For clarity, we often omit the normalization factors from
quantum states.

The following lemma shows that we can efficiently create a `discrete quantum Gaussian state' of
width $r$ as long as $r$ is large enough compared with $\lambda_n(L)$. It can be seen
as the quantum analogue of Lemma~\ref{lem:bootstrap}. The assumption that $L \subseteq \Z^n$
is essentially without loss of generality since a lattice with rational coordinates can
always be rescaled so that $L \subseteq \Z^n$.

\begin{lemma}\label{lem:quantumstate}
There exists an efficient quantum algorithm that, given an $n$-dimensional lattice $L \subseteq \Z^n$
and $r > 2^{2n} \lambda_n(L)$, outputs a state that is within $\ell_2$ distance $2^{-\Omega(n)}$ of
the normalized state corresponding to
\begin{align}\label{eq:requiredquantumstate}
 \sum_{\vec x \in L} \sqrt{\rho_r(\vec x)} \ket{\vec x} =  \sum_{\vec x \in L} \rho_{\sqrt{2}r}(\vec x) \ket{\vec x}.
\end{align}
\end{lemma}
\begin{proof}
We start by creating the `one-dimensional Gaussian state'
\begin{align}\label{eq:onedimdgaussian}
 \sum_{x=-\sqrt{n}r}^{\sqrt{n}r} e^{-\pi (x/(\sqrt{2}r))^2} \ket{x}.
\end{align}
This state can be created efficiently using a technique by Grover and Rudolph~\cite{GroverCreatingSuper}
who show that in order to create such a state, it suffices to be able to
compute for any $a,b \in \{-\sqrt{n}r,\ldots,\sqrt{n}r\}$
the sum $\sum_{x=a}^{b} e^{-\pi (x/r)^2}$
to within good precision. This can be done using the same standard techniques
used in sampling from the normal distribution.

Repeating the procedure described above $n$ times, creates a system whose
state is the $n$-fold tensor product of the state in Eq.~\eqref{eq:onedimdgaussian},
which can be written as
\begin{align*}
 \sum_{\vec x \in \{-\sqrt{n}r,\ldots,\sqrt{n}r\}^n} \rho_{\sqrt{2} r}(\vec x) \ket{\vec x}.
\end{align*}
Since $\Z^n \cap \sqrt{n} r B_n \subseteq \{-\sqrt{n}r,\ldots,\sqrt{n}r\}^n$,
Lemma~\ref{lem:bana2} implies that this state is within $\ell_2$ distance $2^{-\Omega(n)}$ of
\begin{align}\label{eq:ndimdgaussian}
 \sum_{\vec x \in \Z^n} \rho_{\sqrt{2}r}(\vec x) \ket{\vec x}
\end{align}
and hence for our purposes we can assume that we have generated
the state in Eq.~\eqref{eq:ndimdgaussian}.

Next, using the LLL basis reduction algorithm \cite{LLL}, we obtain a basis for $L$ of length at most $2^n
\lambda_n(L)$ and let $\calP(L)$ be the parallelepiped generated by this
basis. We now compute in a new register $\vec x ~\mod~ \calP(L)$ and measure it.
Let $\vec y \in \calP(L)$ denote the result and note that $\|\vec y\| \le \diam(\calP(L)) \le n 2^n \lambda_n(L)$.
The state we obtain after the measurement is
$$
 \sum_{\vec x \in L+\vec y} \rho_{\sqrt{2} r}(\vec x) \ket{\vec x}.
$$
Finally, we subtract $\vec y$ from our register, and obtain
$$
 \sum_{\vec x \in L} \rho_{\sqrt{2} r}(\vec x +\vec y) \ket{\vec x}.
$$

Our goal is to show that this state is within $\ell_2$ distance $2^{-\Omega(n)}$ of the
one in Eq.~\eqref{eq:requiredquantumstate}. First, by Lemma~\ref{lem:bana2}, all but
an exponentially small part of the $\ell_2$ norm of the state in Eq.~\eqref{eq:requiredquantumstate}
is concentrated on points of norm at most $\sqrt{n} \cdot r$. So consider any $\vec x \in L$
with $\|\vec x\| \le \sqrt{n} \cdot r$. The amplitude squared given to it in Eq.~\eqref{eq:requiredquantumstate}
is $\rho_{r}(\vec x) / \rho_{r}(L)$.
By Lemma~\ref{lem:psf}, the denominator is
$\rho_{r}(L) = \det(L^*) \cdot r^n \rho_{1/r}(L^*) \ge \det(L^*) \cdot r^n$
and hence the amplitude squared is at most
$\rho_{r}(\vec x) / (\det(L^*) \cdot r^n) = \det(L) \nu_{r}(\vec x)$.

On the other hand, the amplitude squared given to $\vec x$
by our procedure is $\rho_{r}(\vec x +\vec y) / \rho_r(L+\vec y)$.
By Lemma~\ref{lem:psf}, the denominator is
$$\rho_{r}(L+\vec y) = \det(L^*) \cdot r^n \sum_{\vec z \in L^*} e^{2 \pi i \ip{\vec z, \vec y}} \rho_{1/r}(\vec z)
   \le (1+2^{-\Omega(n)}) \det(L^*) \cdot r^n.$$
To obtain this inequality, first note that by the easy part of Lemma~\ref{lem:transference},
$\lambda_1(L^*) \ge 1/ \lambda_n(L) > \sqrt{n}/r$,
and then apply Lemma~\ref{lem:bana2}.
Moreover, by Claim~\ref{clm:gaussianperturb}, the numerator is at least
$(1-2^{-\Omega(n)}) \rho_{r}(\vec x)$.
Hence, the amplitude squared given to $\vec x$ is at least
$(1-2^{-\Omega(n)}) \det(L) \nu_{r}(\vec x)$, as required.
\end{proof}

For a lattice $L$ and a positive integer $R$, we denote by $L/R$
the lattice obtained by scaling down $L$ by a factor of $R$.
The following technical claim follows from the fact that almost all the mass of $\rho$
is on points of norm at most $\sqrt{n}$.\footnote{The proof originally provided in the published journal version of this paper (J. ACM, 56(6), 34, 2009) contained a bug. The bug was found and fixed by the author in summer 2023, and also independently by Yilei Chen, Zihan Hu, Qipeng Liu, Han Luo, and Yaxin Tu in late 2023. The current statement of the claim is identical to the original one. However, the condition $\lambda_1(L) > 2\sqrt{n}$ is no longer used in the proof.}

\begin{claim}\label{clm:quantum_approx}
Let $R \ge 1$ be an integer and $L$ be an $n$-dimensional lattice satisfying $\lambda_1(L) > 2\sqrt{n}$.
Let $\calP(L)$ be some basic parallelepiped of $L$.
Then, the $\ell_2$ distance between the normalized quantum states corresponding to
\begin{align*}
 \ket{\vartheta_1} &= \sum_{\vec x \in L/R, \|\vec x\| < \sqrt{n}} \rho(\vec x) \ket{ \vec x ~\mod~ \calP(L)}, \text{\qquad and} \\
 \ket{\vartheta_2} &= \sum_{\vec x \in L/R} \rho(\vec x) \ket{ \vec x ~\mod~ \calP(L)} =
 \sum_{\vec x \in L/R \cap \calP(L)} \sum_{\vec y \in L} \rho(\vec x - \vec y) \ket{\vec x}
\end{align*}
is $2^{-\Omega(n)}$.
\end{claim}
\begin{proof}
Let $Z$ be the $\ell_2$ norm of $\ket{\vartheta_2}$. 
In the following we show that the $\ell_2$ distance between $\ket{\vartheta_1}$ and
$\ket{\vartheta_2}$ is at most $2^{-\Omega(n)} Z$. This is enough to establish that the $\ell_2$ distance between the
normalized quantum states corresponding to $\ket{\vartheta_1}$ and $\ket{\vartheta_2}$ is exponentially small.

Define the $2n$-dimensional lattice 
\begin{align*}
M &= \{ (\vec x_1, \vec x_2) \in L/R \times L/R ~|~  \vec x_1 = \vec x_2 ~\mod~ \calP(L) \} \\
  &= \bigcup_{\vec x \in L/R \cap \calP(L)} (\vec x+L) \times (\vec x+L) \; ,
\end{align*}
where the union is of disjoint sets. It follows that $\rho(M) = \sum_{\vec x \in L/R \cap \calP(L)} \rho(\vec x + L)^2 = Z^2$.

Next, notice that 
\[
\ket{\vartheta_2} - \ket{\vartheta_1} = 
\sum_{\vec x \in L/R \cap \calP(L)} \rho((\vec x + L) \setminus \sqrt{n} B_n) \ket{\vec x} \; .
\]
Therefore, $\|\ket{\vartheta_2} - \ket{\vartheta_1}\|^2 = \rho(M')$ where $M'$ is the subset of $M$ containing all pairs $(\vec x_1, \vec x_2)$ such that both $\|\vec x_1\| \ge \sqrt{n}$ and $\|\vec x_2\| \ge \sqrt{n}$. 
Since $M' \subset M \setminus \sqrt{2n} B_n$, we can apply Lemma~\ref{lem:bana2} (with $r=1$) and obtain that $\rho(M') < 2^{-4n} \rho(M) = 2^{-4n} Z^2$, as desired. 
\end{proof}

We now prove the main lemma of this subsection.

\begin{lemma}[Second part of iterative step]\label{lem:quantum_part}
There exists an efficient quantum algorithm that, given any $n$-dimensional lattice $L$, a number
$d < \lambda_1(L^*)/2$, and an oracle that solves $\CVP_{L^*,d}$, outputs a sample
from $D_{L,\sqrt{n}/(\sqrt{2}d)}$.
\end{lemma}
\begin{proof}
By scaling, we can assume without loss of generality that $d=\sqrt{n}$.
Let $R \ge 2^{3n}\lambda_n(L^*)$ be a large enough integer.
We can assume that $\log R$ is polynomial in the input size (since such an $R$ can be computed in polynomial
time given the lattice $L$). Our first step is to create a state exponentially close to
\begin{equation}\label{eq:first_quantum_state}
\sum_{\vec x \in L^*/R \cap \calP(L^*)} \sum_{\vec y \in L^*} \rho(\vec x - \vec y) \ket{\vec x}.
\end{equation}
This is a state on $n \log R$ qubits, a number that is polynomial in the input size. To do so,
we first use Lemma~\ref{lem:quantumstate} with $r=1/\sqrt{2}$ and the lattice $L^*/R$ to create the state
$$ \sum_{\vec x \in L^*/R} \rho(\vec x) \ket{\vec x}.$$
By Lemma~\ref{lem:bana2}, this is exponentially close to
$$ \sum_{\vec x \in L^*/R, \|\vec x\| < \sqrt{n}} \rho(\vec x) \ket{\vec x}.$$
Next, we compute $\vec x ~\mod~ \calP(L^*)$ in a new register and obtain
$$ \sum_{\vec x \in L^*/R, \|\vec x\| < \sqrt{n}} \rho(\vec x) \ket{\vec x, \vec x ~\mod~ \calP(L^*)}.$$
Using the $\CVP$ oracle, we can recover $\vec x$ from $\vec x ~\mod~ \calP(L^*)$. This allows us to uncompute the first
register and obtain
$$ \sum_{\vec x \in L^*/R, \|\vec x\| < \sqrt{n}} \rho(\vec x) \ket{\vec x ~\mod~ \calP(L^*)}.$$
Using Claim~\ref{clm:quantum_approx}, this state is exponentially close to the required state
\eqref{eq:first_quantum_state}.

In the second step, we apply the quantum Fourier transform. First, using the natural mapping between
$L^*/R \cap \calP(L^*)$ and $\Z_{R}^n$, we can rewrite \eqref{eq:first_quantum_state} as
$$ \sum_{\vec s \in \Z_{R}^n} \sum_{\vec r \in \Z^n} \rho(L^* \vec s /R - L^* \vec r) \ket{\vec s} .$$
We now apply the quantum Fourier transform on $\Z_{R}^n$. We obtain a state in which the amplitude of $\ket{\vec t}$
for $\vec t \in \Z_{R}^n$ is proportional to
\begin{align*}
   & \sum_{\vec s\in \Z^n_{R}} \sum_{\vec r \in \Z^n} \rho(L^* \vec s / R - L^* \vec r)
   \exp(2\pi i \ip{\vec s,\vec t}/R) \\
   & \qquad = \sum_{\vec s\in \Z^n} \rho(L^* \vec s / R ) \exp(2\pi i \ip{\vec s,\vec t}/R) \\
   &\qquad = \sum_{\vec x\in L^*/R} \rho(\vec x) \exp(2\pi i \ip{(L^*)^{-1}\vec x,\vec t}) \\
   &\qquad = \sum_{\vec x\in L^*/R} \rho(\vec x) \exp(2\pi i \ip{\vec x,L \vec t}) \\
   &\qquad = \det(R L) \sum_{\vec y\in R L} \rho(\vec y - L \vec t)
\end{align*}
where the last equality follows from Lemma~\ref{lem:psf} and Eq.~\eqref{eq:fourier_shift_2}. Hence, the resulting state can
be equivalently written as
$$ \sum_{\vec x \in \calP(RL) \cap L} \sum_{\vec y \in RL} \rho(\vec y - \vec x) \ket{\vec x}.$$
Notice that $\lambda_1(RL) = R \lambda_1(L) \ge R / \lambda_n(L^*) \ge 2^{3n}$. Hence, we can apply Claim
\ref{clm:quantum_approx} to the lattice $RL$, and obtain that this state is exponentially close to
$$ \sum_{\vec x \in L, \|\vec x\| < \sqrt{n}} \rho(\vec x) \ket{\vec x ~\mod~ \calP(RL)}.$$
We measure this state and obtain $\vec x ~\mod~ \calP(RL)$ for some vector $\vec x$ with $\|\vec x\| < \sqrt{n}$. Since
$\vec x ~\mod~ \calP(RL)$ is within $\sqrt{n}$ of the lattice $RL$, and $\lambda_1(RL) \ge 2^{3n}$, we can recover $\vec x$
by using, say, Babai's nearest plane algorithm~\cite{Babai86}. The output of the algorithm is $\vec x$.

We claim that the distribution of $\vec x$ is exponentially close to $D_{L,1/\sqrt{2}}$. Indeed, the probability of
obtaining any $\vec x \in L, \| \vec x\| < \sqrt{n}$ is proportional to $\rho(\vec x)^2 = \rho_{1/\sqrt{2}}(\vec x)$.
It remains to notice that by Lemma~\ref{lem:bana2}, all but an exponentially small fraction of the probability
distribution $D_{L,1/\sqrt{2}}$ is on points of norm less than $\sqrt{n}$.
\end{proof}

\subsection{Standard lattice problems}\label{ssec:standardlatticeproblems}

We now complete the proof of the main theorem by reducing the standard lattice problems $\GapSVP$ and $\SIVP$
to $\DGS$. We start with $\SIVP$. The basic idea of the reduction is simple: we call the $\DGS$ oracle
enough times. We show that with high probability, there are $n$ short linearly independent
vectors among the returned vectors. We prove this by using the following
lemma, which appeared in the preliminary version of \cite{MicciancioR04}. We include the
proof since only a proof sketch was given there.

\begin{lemma}\label{lem:hyperplane}
Let $L$ be an $n$-dimensional lattice and let $r$ be such that $r \ge \sqrt{2} \eta_\epsilon(L)$ where $\epsilon \le
\frac{1}{10}$. Then for any subspace $H$ of dimension at most $n-1$ the probability that $\vec{x}\notin H$ where $\vec{x}$ is
chosen from $D_{L,r}$ is at least $\frac{1}{10}$.
\end{lemma}
\begin{proof}
Assume without loss of generality that the vector $(1,0,\ldots,0)$ is orthogonal to $H$.
Using Lemma~\ref{lem:psf},
\begin{align*}
\Exp_{\vec x \sim D_{L,r}}[ \expo{-\pi (x_1 /r )^2} ] &=
\frac{1}{\rho_r(L)} \sum_{\vec x \in L} \expo{ - \pi (\sqrt{2}x_1/r)^2}
  \expo{ - \pi (x_2/r)^2} \cdots \expo{ - \pi (x_n/r)^2} \\
  &= \frac{\det(L^*) \, r^n}{\sqrt{2} \rho_r(L)} \, \sum_{\vec y \in L^*}
       \expo{-\pi (r y_1 / \sqrt{2})^2} \expo{-\pi (r y_2)^2} \cdots \expo{-\pi (r y_n)^2} \\
  &\le \frac{\det(L^*) \, r^n}{\sqrt{2} \rho_r(L)} \, \rho_{\sqrt{2}/r} (L^*) \\
  &\le \frac{\det(L^*) \, r^n}{\sqrt{2} \rho_r(L)} \, (1+\eps).
\end{align*}
By using Lemma~\ref{lem:psf} again, we see that $\rho_r(L) = \det(L^*) \, r^n \rho_{1/r}(L^*) \ge \det(L^*) \, r^n$.
Therefore, the expectation above is at most
$\frac{1}{\sqrt{2}} \,(1+\eps) < 0.9$ and the lemma follows.
\end{proof}

\begin{corollary}\label{cor:nindependentvectors}
Let $L$ be an $n$-dimensional lattice and let $r$ be such that $r \ge \sqrt{2} \eta_\epsilon(L)$ where $\eps \le \frac{1}{10}$. Then, the
probability that a set of $n^2$ vectors chosen independently from $D_{L,r}$ contains no $n$ linearly independent vectors is
exponentially small.
\end{corollary}
\begin{proof}
Let $\vec x_1,\ldots,\vec x_{n^2}$ be $n^2$ vectors chosen independently from $D_{L,r}$. For $i=1,\ldots,n$,
let $B_i$ be the event that
$$\dim {\rm span}(\vec x_1,\ldots,\vec x_{(i-1)n}) = \dim  {\rm span}(\vec x_1,\ldots,\vec x_{in}) < n.$$
Clearly, if none of the $B_i$'s happens,
then $\dim {\rm span}(\vec x_1,\ldots,\vec x_{n^2})=n$. Hence, in order to complete the proof it suffices to
show that for all $i$, $\Pr[B_i] \le 2^{-\Omega(n)}$. So fix some $i$, and let us condition on some fixed
choice of $\vec x_1,\ldots,\vec x_{(i-1)n}$ such that $\dim {\rm span}(\vec x_1,\ldots,\vec x_{(i-1)n}) < n$.
By Lemma~\ref{lem:hyperplane}, the probability that
$$\vec x_{(i-1)n+1},\ldots,\vec x_{in} \in \dim {\rm span}(\vec x_1,\ldots,\vec x_{(i-1)n})$$
is at most $(9/10)^n = 2^{-\Omega(n)}$.
This implies that $\Pr[B_i] \le 2^{-\Omega(n)}$, as required.
\end{proof}

In the following lemma we give the reduction from $\SIVP$ (in fact, $\GIVP$) to $\DGS$.
It shows that under the assumptions of Theorem~\ref{thm:mainthmmm},
there exists an efficient quantum algorithm for $\GIVP_{2\sqrt{2}n \eta_\eps(L)/\alpha}$.
By Lemma~\ref{lem:lambdan}, this algorithm also solves $\SIVP_{\tilde{O}(n / \alpha)}$.

\begin{lemma}\label{lem:givptodgs}
For any $\eps=\eps(n) \le \frac{1}{10}$ and any $\varphi(L) \ge \sqrt{2} \eta_\epsilon(L)$,
there is a polynomial time reduction from $\GIVP_{2 \sqrt{n} \varphi}$ to $\DGS_\varphi$.
\end{lemma}

\begin{proof}
As mentioned above, the idea of the reduction is to simply call the $\DGS$ oracle in an attempt to find
$n$ short linearly independent vectors. One technical complication is that the function $\varphi$ is
not necessarily efficiently computable, and hence we do not know which parameter $r$ to give the $\DGS$ oracle.
The solution is easy: we just try many values of $r$ and take the shortest set of $n$
linearly independent vectors found.

We now present the reduction in detail. The input to the reduction is a lattice $L$.
We first apply the LLL algorithm \cite{LLL} to obtain $n$ linearly independent vectors of length at most $2^n
\lambda_n(L)$. Let $S$ denote the resulting set, and let $\tilde{\lambda_n}$ be the length of the longest
vector in $S$. By construction we have $\lambda_n(L) \le \tilde{\lambda_n} \le 2^n \lambda_n(L)$.
For each $i \in \{0,\ldots,2n\}$
call the $\DGS$ oracle $n^2$ times with the pair $(L,r_i)$ where $r_i = \tilde{\lambda_n} 2^{-i}$,
and let $S_i$ be the resulting set of vectors.
At the end, look for a set of $n$ linearly independent vectors in each of $S,S_0,S_1,\ldots,S_{2n}$,
and output the shortest set found.

We now prove correctness. If $\varphi(L) \ge \tilde{\lambda_n}$ then $S$ is already
shorter than $2 \sqrt{n} \varphi(L)$ and so we are done. Otherwise, let $i \in \{0,\ldots,2n\}$ be such that
$\varphi(L) < r_i \le 2 \varphi(L)$. Such an $i$ must exist by Claim~\ref{clm:lowerboundsmoothing}.
By Corollary~\ref{cor:nindependentvectors}, $S_i$ contains $n$ linearly independent vectors with
probability exponentially close to $1$. Moreover, by Lemma~\ref{lem:bana2}, all vectors in $S_i$ are of
length at most $r_i \sqrt{n} \le 2 \sqrt{n} \varphi(L)$ with probability exponentially close to $1$. Hence, our reduction
outputs a set of $n$ linearly independent vectors of length at most $2 \sqrt{n} \varphi(L)$, as required.
\end{proof}

We now present the reduction from $\GapSVP$ to $\DGS$. We first define the decision version of the closest vector problem ($\GapCVP$) and a
slight variant of it.

\begin{definition}
An instance of $\GapCVP_\gamma$ is given by an $n$-dimensional lattice $L$, a vector $\vec t$, and a number $d>0$.
In $\YES$ instances, $\dist(\vec t,L) \leq d$, whereas
in $\NO$ instances, $\dist(\vec t,L) > \gamma(n) \cdot d$.
\end{definition}

\begin{definition}
An instance of $\GapCVP'_\gamma$ is given by an $n$-dimensional lattice $L$, a vector $\vec t$, and a number $d>0$.
In $\YES$ instances, $\dist(\vec t,L) \leq d$.
In $\NO$ instances, $\lambda_1(L) > \gamma(n) \cdot d$ and $\dist(\vec t,L) > \gamma(n) \cdot d$.
\end{definition}

\noindent
In~\cite{GMSS99} it is shown that for any $\gamma=\gamma(n) \ge 1$, there is a polynomial time reduction
from $\GapSVP_\gamma$ to $\GapCVP'_\gamma$ (see also Lemma 5.22 in \cite{MicciancioR04}). Hence, it suffices
to show a reduction from $\GapCVP'$ to $\DGS$. This reduction is given in the following lemma.
By using Lemma~\ref{lem:lambda1}, we obtain that under the assumptions of Theorem~\ref{thm:mainthmmm}
there exists an efficient quantum algorithm for $\GapCVP'_{O(n/\alpha)}$ (and hence also for
$\GapSVP_{O(n/\alpha)}$).

\begin{lemma}\label{lem:svptodgs}
For any $\gamma=\gamma(n) \ge 1$,
there is a polynomial time reduction from $\GapCVP'_{100\sqrt{n} \cdot \gamma(n)}$ to
$\DGS_{\sqrt{n} \gamma(n) / \lambda_1(L^*)}$.
\end{lemma}
\begin{proof}
The main component in our reduction is the $\NP$ verifier for $\coGapCVP$ shown in \cite{AharonovR04}.
In more detail, \cite{AharonovR04} present an efficient algorithm, call it $\calV$,
whose input consists of an $n$-dimensional lattice $L$, a vector $\vec t$, a number $d>0$, and a sequence of vectors
$\vec w_1,\ldots, \vec w_N$ in $L^*$ for some $N=\poly(n)$. When $\dist(\vec t, L) \leq d$, the
algorithm is guaranteed to reject. When $\dist(\vec t, L) > 100 \sqrt{n} d$, and
$\vec w_1,\ldots, \vec w_N$ are chosen from the distribution $D_{L^*,1/(100d)}$, then the algorithm
accepts with probability exponentially close to $1$.

The input to the reduction is an $n$-dimensional lattice $L$, a vector $\vec t$, and a number $d>0$.
We call the $\DGS$ oracle $N$ times with the lattice $L^*$ and the value $\frac{1}{100 d}$ to obtain vectors
$\vec w_1,\ldots, \vec w_N \in L^*$. We then apply $\calV$ with $L$, $\vec t$, $d$, and the
vectors $\vec w_1,\ldots, \vec w_N$. We accept if and only if $\calV$ rejects.

To prove correctness, notice first that in the case of a $\YES$ instance,
$\dist(\vec t,L) \le d$, and hence $\calV$ must reject (irrespective of the $\vec w$'s).
In the case of a $\NO$ instance we have that
$\frac{1}{100 d} > \sqrt{n} \gamma(n) / \lambda_1(L)$, and hence
$\vec w_1,\ldots, \vec w_N$ are guaranteed to be valid samples from $D_{L^*, 1/(100d)}$.
Moreover, $\dist(\vec t,L) > 100\sqrt{n} \gamma(n) d \ge 100 \sqrt{n} d$, and hence $\calV$ accepts with probability
exponentially close to $1$.
\end{proof}

\section{Variants of the $\LWE$ problem}\label{sec:reductions}

In this section, we consider several variants of the $\LWE$ problem. Through a sequence of elementary reductions, we
prove that all problems are as hard as $\LWE$. The results of this section are summarized in Lemma
\ref{lem:summary_reductions}.

\begin{lemma}[Average-case to Worst-case]
Let $n,p \ge 1$ be some integers and $\chi$ be some distribution on $\Z_p$.
Assume that we have access to a distinguisher $W$ that distinguishes $A_{\vec{s},\chi}$ from $U$ for a non-negligible fraction of
all possible $\vec{s}$. Then there exists an efficient algorithm $W'$ that for {\em all} $\vec{s}$ accepts with probability
exponentially close to $1$ on inputs from $A_{\vec{s},\chi}$ and rejects with probability exponentially close to $1$ on
inputs from $U$.
\end{lemma}
\begin{proof}
The proof is based on the following transformation. For any $\vec{t} \in \Z_p^n$ consider the function
$f_{\vec{t}}:\Z_p^n \times \Z_p \rightarrow \Z_p^n \times \Z_p$ defined by
$$f_{\vec{t}}(\vec{a},b)=(\vec{a}, b+ \ip{\vec{a},\vec{t}}).$$
It is easy to see that this function transforms the distribution $A_{\vec{s},\chi}$ into $A_{\vec{s}+\vec{t},\chi}$.
Moreover, it transforms the uniform distribution $U$ into itself.

Assume that for $n^{-c_1}$ of all possible $\vec{s}$, the acceptance probability of $W$ on inputs from
$A_{\vec{s},\chi}$ and on inputs from $U$ differ by at least $n^{-c_2}$. We construct $W'$ as follows. Let $R$ denote
the unknown input distribution. Repeat the following $n^{c_1+1}$ times. Choose a vector $\vec{t} \in \Z_p^n$ uniformly
at random. Then estimate the acceptance probability of $W$ on $U$ and on $f_{\vec{t}}(R)$ by calling $W$ $O(n^{2c_2+1})$
times on each of the input distributions. By the Chernoff bound, this allows us to obtain an estimate that with probability exponentially
close to $1$ is within $\pm n^{-c_2}/8$ of the true acceptance probabilities. If the two estimates differ by more than
$n^{-c_2}/2$ then we stop and decide to accept. Otherwise we continue. If the procedure ends without accepting, we
reject.

We now prove that $W'$ distinguishes $A_{\vec{s},\chi}$ from $U$ for all $\vec{s}$. First, we claim that when $R$ is
$U$, the acceptance probability of $W'$ is exponentially close to $0$. Indeed, in this case, $f_{\vec{t}}(U)=U$ and
therefore the two estimates that $W'$ performs are of the same distribution. The probability that the estimates differ
by more than $n^{-c_2}/2 > 2 \cdot n^{-c_2}/8$ is exponentially small. Next, consider the case that $R$ is
$A_{\vec{s},\chi}$ for some $\vec{s}$. In each of the $n^{c_1+1}$ iterations, we are considering the distribution
$f_{\vec{t}}(A_{\vec{s},\chi})=A_{\vec{s}+\vec{t},\chi}$ for some uniformly chosen $\vec{t}$. Notice that the
distribution of $\vec{s}+\vec{t}$ is uniform on $\Z_p^n$. Hence, with probability exponentially close to $1$, in one of
the $n^{c_1+1}$ iterations, $\vec{t}$ is such that the acceptance probability of $W$ on inputs from
$A_{\vec{s}+\vec{t},\chi}$ and on inputs from $U$ differ by at least $n^{-c_2}$. Since our estimates are within $\pm
n^{-c_2}/8$, $W'$ accepts with probability exponentially to $1$.
\end{proof}

\begin{lemma}[Decision to Search]\label{lem:decision_search}
Let $n \ge 1$ be some integer, $2 \le p \le \poly(n)$ be a prime, and $\chi$ be some distribution on $\Z_p$. Assume that we have
access to procedure $W$ that for all $\vec{s}$ accepts with probability exponentially close to $1$ on inputs from
$A_{\vec{s},\chi}$ and rejects with probability exponentially close to $1$ on inputs from $U$. Then, there exists an
efficient algorithm $W'$ that, given samples from $A_{\vec{s},\chi}$ for some $\vec{s}$, outputs $\vec{s}$ with probability
exponentially close to $1$.
\end{lemma}
\begin{proof}
Let us show how $W'$ finds $s_1 \in \Z_p$, the first coordinate of $\vec{s}$. Finding the other coordinates is similar.
For any $k \in \Z_p$, consider the following transformation. Given a pair $(\vec{a},b)$ we output the pair
$(\vec{a}+(l,0,\ldots,0), b+ l\cdot k)$ where $l \in \Z_p$ is chosen uniformly at random. It is easy to see that this
transformation takes the uniform distribution into itself. Moreover, if $k=s_1$ then this transformation also takes
$A_{\vec{s},\chi}$ to itself. Finally, if $k\neq s_1$ then it takes $A_{\vec{s},\chi}$ to the uniform distribution
(note that this requires $p$ to be prime). Hence, using $W$, we can test whether $k=s_1$. Since there are only $p <
\poly(n)$ possibilities for $s_1$ we can try all of them.
\end{proof}

\begin{lemma}[Discrete to Continuous]\label{lem:learning_discretize}
Let $n,p \ge 1$ be some integers, let $\phi$ be some probability density function on $\T$,
and let $\bar{\phi}$ be its discretization to $\Z_p$. Assume
that we have access to an algorithm $W$ that solves $\LWE_{p,\bar{\phi}}$. Then, there exists an efficient algorithm $W'$ that solves
$\LWE_{p,\phi}$.
\end{lemma}
\begin{proof}
Algorithm $W'$ simply takes samples from $A_{\vec{s},\phi}$ and discretizes the second element to obtain samples from
$A_{\vec{s},\bar{\phi}}$. It then applies $W$ with these samples in order to find $\vec{s}$.
\end{proof}

By combining the three lemmas above, we obtain

\begin{lemma}\label{lem:summary_reductions}
Let $n \ge 1$ be an integer and $2 \le p \le \poly(n)$ be a prime. Let $\phi$ be some probability density function on $\T$ and let
$\bar{\phi}$ be its discretization to $\Z_p$. Assume that we have access to a distinguisher that distinguishes
$A_{\vec{s},\bar{\phi}}$ from $U$ for a non-negligible fraction of all possible $\vec{s}$. Then, there exists an
efficient algorithm that solves $\LWE_{p,\phi}$.
\end{lemma}

\section{Public Key Cryptosystem}\label{sec:cryptosystem}

We let $n$ be the security parameter of the cryptosystem. Our cryptosystem is parameterized by two integers $m,p$ and a
probability distribution $\chi$ on $\Z_p$. A setting of these parameters that guarantees both security and correctness
is the following. Choose $p \ge 2$ to be some prime number between $n^2$ and $2n^2$ and let $m=(1+\eps)(n+1)\log p$ for some
arbitrary constant  $\eps>0$. The probability
distribution $\chi$ is taken to be $\bar{\Psi}_{\alpha(n)}$ for $\alpha(n) = o(1/(\sqrt{n} \log n))$, i.e., $\alpha(n)$
is such that $\lim_{n \to \infty} \alpha(n) \cdot \sqrt{n} \log n = 0$. For example, we can choose $\alpha(n) = 1/(\sqrt{n}
\log^2 n)$. In the following description, all additions are performed in $\Z_p$, i.e., modulo $p$.

 \begin{itemize}
 \item
{\bf Private key:} Choose $\vec{s} \in \Z_p^n$ uniformly at random. The private key is $\vec{s}$.
 \item
{\bf Public Key:} For $i=1,\ldots,m$, choose $m$ vectors $\vec{a}_1,\ldots,\vec{a}_m \in \Z_p^n$ independently from the
uniform distribution. Also choose elements $e_1,\ldots,e_m \in \Z_p$ independently according to $\chi$. The public key
is given by $(\vec{a}_i,b_i)_{i=1}^m$ where $b_i=\ip{\vec{a}_i,\vec{s}}+e_i$.
 \item
{\bf Encryption:} In order to encrypt a bit we choose a random set $S$ uniformly among all $2^m$ subsets of $[m]$. The encryption is $(\sum_{i\in S}
\vec{a}_i, \sum_{i\in S} b_i)$ if the bit is 0 and $(\sum_{i\in S} \vec{a}_i, \floor{\frac{p}{2}} + \sum_{i\in S} b_i)$
if the bit is 1.
 \item
{\bf Decryption:} The decryption of a pair $(\vec{a},b)$ is $0$ if $b-\ip{\vec{a},\vec{s}}$ is closer to $0$ than to
$\floor{\frac{p}{2}}$ modulo $p$. Otherwise, the decryption is $1$.
\end{itemize}

Notice that with our choice of parameters, the public key size is $O(mn \log p)=\tilde{O}(n^2)$ and the encryption
process increases the size of a message by a factor of $O(n \log p)=\tilde{O}(n)$. In fact, it is possible to reduce
the size of the public key to $O(m \log p) = \tilde{O}(n)$ by the following idea of Ajtai \cite{AjtaiHardLattices}.
Assume all users of the cryptosystem share some fixed (and trusted) random choice of $\vec{a}_1,\ldots,\vec{a}_m$. This can be
achieved by, say, distributing these vectors as part of the encryption and decryption software. Then the public key
need only consist of $b_1,\ldots,b_m$. This modification does not affect the security of the cryptosystem.

We next prove that under a certain condition on $\chi$, $m$, and $p$, the probability of decryption error is small. We
later show that our choice of parameters satisfies this condition. For the following two lemmas we need to introduce
some additional notation. For a distribution $\chi$ on $\Z_p$ and an integer $k \ge 0$, we define $\chi^{\star k}$ as
the distribution obtained by summing together $k$ independent samples from $\chi$, where addition is performed in
$\Z_p$ (for $k=0$ we define $\chi^{\star 0}$ as the distribution that is constantly $0$). For a probability
distribution $\phi$ on $\T$ we define $\phi^{\star k}$ similarly. For an element $a \in \Z_p$ we define $|a|$ as the
integer $a$ if $a \in \{0,1,\ldots, \floor{\frac{p}{2}} \}$ and as the integer $p-a$ otherwise. In other words, $|a|$
represents the distance of $a$ from $0$. Similarly, for $x \in \T$, we define $|x|$ as $x$ for $x\in [0,\frac{1}{2}]$
and as $1-x$ otherwise.

\begin{lemma}[Correctness]
Let $\delta > 0$. Assume that for any $k \in \{0,1,\ldots,m\}$, $\chi^{\star k}$ satisfies that
$$\Pr_{e\sim \chi^{\star k}}\Big[|e| < \Big\lfloor \frac{p}{2} \Big\rfloor/2 \Big] > 1 - \delta.$$
Then, the probability of decryption error is at most $\delta$. That is, for any bit $c \in \{0,1\}$,
if we use the protocol above to choose private and public keys, encrypt $c$, and then decrypt the result,
then the outcome is $c$ with probability at least $1-\delta$.
\end{lemma}
\begin{proof}
Consider first an encryption of $0$. It is given by $(\vec{a},b)$ for $\vec{a}=\sum_{i\in S} \vec{a}_i$ and
$$b= \sum_{i\in S} b_i = \sum_{i\in S} \ip{\vec{a}_i,\vec{s}}+e_i = \ip{\vec{a},\vec{s}} + \sum_{i\in S} e_i.$$
Hence, $b -\ip{\vec{a},\vec{s}}$ is exactly $\sum_{i\in S} e_i$. The distribution of the latter is $\chi^{\star |S|}$.
According to our assumption, $|\sum_{i\in S} e_i|$ is less than $\lfloor \frac{p}{2} \rfloor/2$ with probability at
least $1-\delta$. In this case, it is closer to $0$ than to $\lfloor \frac{p}{2} \rfloor$ and therefore the decryption
is correct. The proof for an encryption of $1$ is similar.
\end{proof}

\begin{claim}
For our choice of parameters it holds that for any $k\in \{0,1,\ldots, m\}$,
$$ \Pr_{e \sim \bar{\Psi}_\alpha^{\star k}} \Big[|e| < \Big\lfloor \frac{p}{2} \Big\rfloor/2 \Big] > 1 - \delta(n)$$
for some negligible function $\delta(n)$.
\end{claim}
\begin{proof}
A sample from $\bar{\Psi}_\alpha^{\star k}$ can be obtained by sampling $x_1,\ldots,x_k$ from $\Psi_\alpha$ and
outputting $\sum_{i=1}^k \round{p x_i} ~\mod ~ p$. Notice that this value is at most $k \le m < p/32$ away from
$\sum_{i=1}^k p x_i ~\mod~ p$. Hence, it is enough to show that $|\sum_{i=1}^k p x_i ~\mod~ p| < p/16$ with high
probability. This condition is equivalent to the condition that $|\sum_{i=1}^k x_i ~\mod~ 1| < 1/16$. Since $\sum_{i=1}^k x_i ~\mod~ 1$ is distributed as $\Psi_{\sqrt{k} \cdot \alpha}$, and $\sqrt{k} \cdot \alpha =
o(1/\sqrt{\log n})$, the probability that $|\sum_{i=1}^k x_i ~\mod~ 1| < 1/16$ is $1-\delta(n)$ for some negligible function
$\delta(n)$.
\end{proof}

In order to prove the security of the system, we need the following
special case of the leftover hash lemma that appears in \cite{ImpagliazzoZ89}.
We include a proof for completeness.
\begin{claim}\label{subset_sum_uniform}
Let $G$ be some finite Abelian group and let $l$ be some integer. For any $l$ elements
$g_1,\ldots,g_l \in G$ consider the statistical distance between the uniform distribution on
$G$ and the distribution given by the sum of a random subset of $g_1,\ldots,g_l$.
Then the expectation of this statistical distance over a uniform choice of $g_1,\ldots,g_l \in G$
is at most $\sqrt{|G|/2^l}$. In particular, the probability that this statistical distance
is more than $\sqrt[4]{|G|/2^l}$ is at most $\sqrt[4]{|G|/2^l}$.
\end{claim}
\begin{proof}
For a choice $\vec{g}=(g_1,\ldots,g_l)$ of $l$ elements from $G$, let $P_{\vec{g}}$ be the
distribution of the sum of a random subsets of $g_1,\ldots,g_l$, i.e.,
$$ P_{\vec{g}}(h) = \frac{1}{2^l} \left| \set{ \vec{b} \in \{0,1\}^l}{\textstyle{\sum_i} b_i g_i = h} \right|.$$
In order to show that this distribution is close to uniform, we compute its $\ell_2$ norm, and
note that it is very close to $1/|G|$. From this it will follow that the distribution must
be close to the uniform distribution.
The $\ell_2$ norm of $P_{\vec{g}}$ is given by
\begin{align*}
\sum_{h\in G} P_{\vec{g}}(h)^2 &= \Pr_{\vec{b},\vec{b}'} \left[ \sum b_i g_i = \sum b_i' g_i \right] \\
 &\le \frac{1}{2^l} + \Pr_{\vec{b},\vec{b}'}\left[ \left. \sum b_i g_i = \sum b_i' g_i \right| \vec{b} \neq \vec{b}'\right].
\end{align*}
Taking expectation over $\vec{g}$, and using the fact that for any $\vec{b} \neq \vec{b}'$,
$\Pr_{\vec{g}}[ \sum b_i g_i = \sum b_i' g_i] = 1/|G|$, we obtain that
\begin{align*}
\Exp_{\vec{g}}\left[\textstyle{\sum_h} P_{\vec{g}}(h)^2 \right] \le
   \frac{1}{2^l} + \frac{1}{|G|}.
\end{align*}
Finally, the expected distance from the uniform distribution is
\begin{align*}
\Exp_{\vec{g}}\left[\textstyle{\sum_h} \left| P_{\vec{g}}(h) - 1/|G| \right|\right] &\le
   \Exp_{\vec{g}}\left[|G|^{1/2} \left( \textstyle{\sum_h} (P_{\vec{g}}(h)-1/|G|)^2 \right)^{1/2}\right] \\
   &= \sqrt{|G|} \Exp_{\vec{g}}\left[ \left( \textstyle{\sum_h} P_{\vec{g}}(h)^2- 1/|G| \right)^{1/2} \right] \\
   &\le \sqrt{|G|} \left( \Exp_{\vec{g}}\left[\textstyle{\sum_h} P_{\vec{g}}(h)^2 \right] - 1/|G| \right)^{1/2} \\
  &\le  \sqrt{\frac{|G|}{2^l}}.
\end{align*}
\end{proof}

We now prove that our cryptosystem is semantically secure, i.e., that it is hard
to distinguish between encryptions of $0$ and encryptions of $1$. More precisely,
we show that if such a distinguisher exists, then there exists a distinguisher
that distinguishes between $A_{\vec{s},\chi}$ and $U$ for a non-negligible
fraction of all $\vec{s}$. If $\chi = \bar{\Psi}_{\alpha}$ and $p \le \poly(n)$
is a prime, then by Lemma~\ref{lem:summary_reductions}, this also implies an efficient (classical) algorithm
that solves $\LWE_{p,\Psi_\alpha}$. This in turn implies, by Theorem~\ref{thm:mainthmmm},
an efficient quantum algorithm for $\DGS_{\sqrt{2n} \cdot \eta_\eps(L) /\alpha}$.
Finally, by Lemma~\ref{lem:givptodgs} we also obtain an efficient quantum algorithm
for $\SIVP_{\tilde{O}(n / \alpha)}$ and by Lemma~\ref{lem:svptodgs} we obtain an
efficient quantum algorithm for $\GapSVP_{O(n/\alpha)}$.

\begin{lemma}[Security]
For any $\eps>0$ and $m \ge (1+\eps)(n+1)\log p$, if there exists a polynomial time algorithm $W$ that distinguishes between encryptions of $0$ and $1$
then there exists a distinguisher $Z$ that distinguishes between $A_{\vec{s},\chi}$ and $U$ for a non-negligible
fraction of all possible $\vec{s}$.
\end{lemma}
\begin{proof}
Let $p_0(W)$ be the acceptance probability of $W$ on input $((\vec{a}_i,b_i)_{i=1}^m,(\vec{a},b))$ where $(\vec{a},b)$
is an encryption of $0$ with the public key $(\vec{a}_i,b_i)_{i=1}^m$ and the probability is taken over the randomness
in the choice of the private and public keys and over the randomness in the encryption algorithm. We define $p_1(W)$
similarly for encryptions of $1$ and let $p_u(W)$ be the acceptance probability of $W$ on inputs
$((\vec{a}_i,b_i)_{i=1}^m,(\vec{a},b))$ where $(\vec{a}_i,b_i)_{i=1}^m$ are again chosen according to the private and
public keys distribution but $(\vec{a},b)$ is chosen uniformly from $\Z_p^n \times \Z_p$. With this notation, our
hypothesis says that $\abs{p_0(W)-p_1(W)} \ge \frac{1}{n^c}$ for some $c>0$.

We now construct a $W'$ for which $\abs{p_0(W')-p_u(W')} \ge \frac{1}{2n^c}$. By our hypothesis, either
$\abs{p_0(W)-p_u(W)} \ge \frac{1}{2n^c}$ or $\abs{p_1(W)-p_u(W)} \ge \frac{1}{2n^c}$. In the former case we take $W'$
to be the same as $W$. In the latter case, we construct $W'$ as follows. On input
$((\vec{a}_i,b_i)_{i=1}^m,(\vec{a},b))$, $W'$ calls $W$ with $((\vec{a}_i,b_i)_{i=1}^m,(\vec{a},\frac{p-1}{2}+b))$.
Notice that this maps the distribution on encryptions of $0$ to the distribution on encryptions of $1$ and the uniform
distribution to itself. Therefore, $W'$ is the required distinguisher.

For $\vec s \in \Z_p^n$, let $p_0(\vec{s})$ be the probability that $W'$ accepts on input $((\vec{a}_i,b_i)_{i=1}^m,(\vec{a},b))$ where
$(\vec{a}_i,b_i)_{i=1}^m$ are chosen from $A_{\vec{s},\chi}$, and $(\vec{a},b)$ is an encryption of $0$ with the public
key $(\vec{a}_i,b_i)_{i=1}^m$. Similarly, define $p_u(\vec{s})$ to be the acceptance probability of $W'$ where
$(\vec{a}_i,b_i)_{i=1}^m$ are chosen from $A_{\vec{s},\chi}$, and $(\vec{a},b)$ is now chosen uniformly at random from
$\Z_p^n \times \Z_p$. Our assumption on $W'$ says that $|\Exp_{\vec s}[p_0(\vec s)] - \Exp_{\vec s}[p_u(\vec s)]| \ge \frac{1}{2n^c}$.
Define
$$Y = \left\{ \vec{s} ~\left|~ |p_0(\vec{s}) - p_u(\vec{s})| \ge \frac{1}{4n^c} \right. \right\}.$$
By an averaging argument we get that a fraction of at least $\frac{1}{4n^c}$ of the $\vec{s}$ are in $Y$.
Hence, it is enough to show a distinguisher $Z$ that
distinguishes between $U$ and $A_{\vec{s},\chi}$ for any $\vec{s} \in Y$.

In the following we describe the distinguisher $Z$. We are given a distribution $R$ that is either $U$ or
$A_{\vec{s},\chi}$ for some $\vec{s} \in Y$. We take $m$ samples $(\vec{a}_i,b_i)_{i=1}^m$ from $R$. Let
$p_0((\vec{a}_i,b_i)_{i=1}^m)$ be the probability that $W'$ accepts on input $((\vec{a}_i,b_i)_{i=1}^m,(\vec{a},b))$
where the probability is taken on the choice of $(\vec{a},b)$ as an encryption of the bit 0 with the public key
$(\vec{a}_i,b_i)_{i=1}^m$. Similarly, let $p_u((\vec{a}_i,b_i)_{i=1}^m)$ be the probability that $W'$ accepts on input
$((\vec{a}_i,b_i)_{i=1}^m,(\vec{a},b))$ where the probability is taken over the choice of $(\vec{a},b)$ as a uniform
element of $\Z_p^n \times \Z_p$. By applying $W'$ a polynomial number of times, the distinguisher $Z$ estimates both $p_0((\vec{a}_i,b_i)_{i=1}^m)$ and
$p_u((\vec{a}_i,b_i)_{i=1}^m)$ up to an additive error of $\frac{1}{64n^c}$. If the two estimates differ by more than
$\frac{1}{16n^c}$, $Z$ accepts. Otherwise, $Z$ rejects.

We first claim that when $R$ is the uniform distribution, $Z$ rejects with high probability. In this case,
$(\vec{a}_i,b_i)_{i=1}^m$ are chosen uniformly from $\Z_p^n \times \Z_p$. Using Claim~\ref{subset_sum_uniform} with the
group $G=\Z_p^n \times \Z_p$, we obtain that with probability exponentially close to $1$, the distribution on
$(\vec{a},b)$ obtained by encryptions of 0 is exponentially close to the uniform distribution on $\Z_p^n \times \Z_p$.
Therefore, except with exponentially small probability,
 $$|p_0((\vec{a}_i,b_i)_{i=1}^m) - p_u((\vec{a}_i,b_i)_{i=1}^m)| \le 2^{-\Omega(n)}.$$
Hence, our two estimates differ by at most $\frac{1}{32n^c} + 2^{-\Omega(n)}$, and $Z$ rejects.

Next, we show that if $R$ is $A_{\vec{s},\chi}$ for $\vec{s} \in Y$ then $Z$ accepts with probability $1/\poly(n)$.
Notice that $p_0(\vec{s})$ (respectively, $p_u(\vec{s})$) is the average of $p_0((\vec{a}_i,b_i)_{i=1}^m)$
(respectively, $p_u((\vec{a}_i,b_i)_{i=1}^m)$) taken over the choice of $(\vec{a}_i,b_i)_{i=1}^m$ from
$A_{\vec{s},\chi}$. From $|p_0(\vec{s}) - p_u(\vec{s})| \ge \frac{1}{4n^c}$ we obtain by an averaging argument that
 $$ |p_0((\vec{a}_i,b_i)_{i=1}^m) - p_u((\vec{a}_i,b_i)_{i=1}^m)| \ge \frac{1}{8n^c}$$
with probability at least $\frac{1}{8n^c}$ over the choice of $(\vec{a}_i,b_i)_{i=1}^m$ from $A_{\vec{s},\chi}$. Hence,
with probability at least $\frac{1}{8n^c}$, $Z$ chooses such a $(\vec{a}_i,b_i)_{i=1}^m$ and since our estimates are
accurate to within $\frac{1}{64n^c}$, the difference between them is more than $\frac{1}{16n^c}$ and $Z$ accepts.
\end{proof}

\subsection*{Acknowledgments}

I would like to thank Michael Langberg, Vadim Lyubashevsky, Daniele Micciancio, Chris Peikert, Miklos Santha,
Madhu Sudan, and an anonymous referee for useful comments.

\bibliographystyle{abbrv}
\bibliography{qcrypto}

\end{document}